\documentclass[11pt,english]{article}
\usepackage{lmodern}

\usepackage[T1]{fontenc}
\usepackage[latin9]{inputenc}
\usepackage{xcolor}
\usepackage{babel}
\usepackage{verbatim}
\usepackage{float}
\usepackage{amsmath}
\usepackage{amsthm}
\usepackage{amssymb}
\usepackage{graphicx}
\usepackage{geometry}
\usepackage{setspace}
\usepackage[authoryear]{natbib}
\usepackage[unicode=true,pdfusetitle,
 bookmarks=true,bookmarksnumbered=false,bookmarksopen=false,
 breaklinks=false,pdfborder={0 0 0},pdfborderstyle={},backref=false,colorlinks=true]
 {hyperref}
\hypersetup{
 linkcolor=magenta, urlcolor=blue, citecolor=blue, pdfstartview={FitH}, hyperfootnotes=false, unicode=true}

\makeatletter
\theoremstyle{plain}
\newtheorem{thm}{\protect\theoremname}
\theoremstyle{plain}
\newtheorem{prop}{\protect\propositionname}
\theoremstyle{plain}
\newtheorem{lem}{\protect\lemmaname}

\setcitestyle{round}
\usepackage{lscape}
\usepackage{longtable}

\usepackage{graphicx}

\usepackage{tabularx}
\newcolumntype{P}[1]{>{\centering\arraybackslash}p{#1}}
\newcolumntype{Y}{>{\raggedleft\arraybackslash}X}

\usepackage{array,booktabs,ragged2e}
\newcolumntype{C}[1]{>{\centering\arraybackslash}p{#1}}
\newcolumntype{J}[1]{>{\justify\arraybackslash}p{#1}}
\newcolumntype{R}[1]{>{\RaggedLeft\arraybackslash}p{#1}}
\newcolumntype{Q}[1]{>{\columncolor{Gray}\RaggedLeft\arraybackslash}p{#1}}
\newcolumntype{L}[1]{>{\RaggedRight\arraybackslash}p{#1}}
\newcolumntype{G}{@{\extracolsep{0.5cm}}l@{\extracolsep{0pt}}}%
\usepackage{multirow}

\ifdefined\showcaptionsetup
 \PassOptionsToPackage{caption=false}{subfig}
\fi
\usepackage{subfig}
\AtBeginDocument{

}

\makeatother

\providecommand{\lemmaname}{Lemma}
\providecommand{\propositionname}{Proposition}
\providecommand{\theoremname}{Theorem}

\begin{document}
\title{A Robust Similarity Estimator \thanks{The author is grateful for helpful comments from Peter Reinhard Hansen,
Jun Yu, Han Chen, Yijie Fei, Marine Carrasco, Benoit Perron, Ren\'{e}
Garcia, Victoria Zinde-Walsh, Yifan Li, and seminar participants at
the University of Macau (October, 2025), the Virtual Time Series Seminar
(December, 2025), University of Montreal (May, 2026). The author also
thanks participants of the 2025 NUS Quantitative Finance Conference,
2026 annual SoFiE meeting, 2026 Asian Meeting of Econometric Society
for valuable discussions and feedback. \protect \linebreak{}
The author gratefully acknowledges financial support from the Social
Sciences and Humanities Research Council of Canada (SSHRC) through
an Insight Grant (No. 435-2026-1701).}{\normalsize\emph{\medskip{}
}}}
\author{\textbf{Ilya Archakov} \bigskip{}
\\
{\normalsize\emph{York University; Department of Economics\medskip{}
}}}
\date{{\normalsize\emph{\date{}}}}
\maketitle
\begin{abstract}
We analyze a measure of statistical association based on the similarity
of the outcomes of random variables, in both sign and magnitude. Motivated
by its attractive properties, we propose a class of estimators for
the linear correlation coefficient, a sample-average and a maximum-likelihood
version, that operate directly on the Fisher scale and possess
a robust sampling distribution that is invariant over the entire class
of elliptical distributions. Under scale homogeneity, the finite-sample
distribution of the estimator is available in exact form, facilitating
robust inference for correlations even in small samples. The similarity
measure extends naturally to higher dimensions, where it admits an
interpretation as an indicator of joint similarity among multiple
random variables. As empirical applications, we construct robust confidence
intervals for financial correlations using intraday returns and develop
a new specification of a multivariate GARCH model with robust correlation
dynamics.

\bigskip{}
\end{abstract}
{\small\textit{Keywords:}}{\small{} Correlation, Robust Estimation,
Robust Inference, Fisher Transformation, Matrix Logarithm, High Frequency
Data, Multivariate GARCH}{\small\par}

\noindent{\small\textit{JEL Classification:}}{\small{} C13, C30, C38,
C58 \newpage}{\small\par}

\section{Introduction}

Measuring statistical association and dependence between random variables
is a broad topic in statistics with a long history. The most popular
and widely used measure of association is the linear correlation coefficient
(denoted by $\rho$), which naturally emerges in many statistical
and econometric frameworks such as linear regression, multivariate
GARCH models, or network analysis. In the financial econometrics literature,
correlations play a critical role in risk management, hedging, optimal
portfolio allocation, the analysis of systemic risk, etc.

Both estimation and inference for correlations are practically challenging,
especially in settings where the sample size is limited. For example,
under sufficiently mild assumptions, the well-known sample correlation
estimator ($\hat{\rho}$) is a consistent estimator of $\rho$. Although
it is an efficient estimator when observations are independent and
normally distributed, its variance is inflated in the presence of
heavy-tailed data, and the estimator remains notoriously sensitive
to outliers. In addition, the finite sample distribution of $\hat{\rho}$
converges very slowly to its asymptotic limit, with both the shape
and the spread often depending strongly on the properties of the underlying
data. As a result, potential distortions in both the central tendency
and the sampling distribution of $\hat{\rho}$ undermine robust estimation
of the correlation coefficient and complicate statistical inference.

In a series of seminal papers, Ronald A. Fisher proposed a continuous
transformation (now known as the Fisher transformation) for $\hat{\rho}$
that offers several advantages, including variance stabilization and
a symmetric, nearly Gaussian sampling distribution for the transformed
sample correlation, even in relatively small samples (see \citet{Fisher_1921},
\citet{Hotelling_1953}). The rapid convergence to the asymptotic
distribution has made the Fisher transformation popular for conducting
statistical inference about $\rho$, even with limited data. The estimation
and inference, however, remain fragile, as the transformation does
not provide robustness to outliers, and the sampling variance remains
inflated under heavy-tailed data distributions.

The lack of robustness of the sample correlation is commonly addressed
through the use of alternative correlation measures that are, by construction,
insensitive to extreme observations. Notable examples include the
Quadrant estimator (\citet{Greiner_1909}, \citet{Kendall_1949},
\citet{Blomqvist_1950}) and the Kendall rank correlation coefficient
(\citet{esscher1924method}, \citet{Kendall_1938}). These measures
are based on the relative number of concordant and discordant pairs
of observations and are therefore intrinsically robust to outliers.
Under mild assumptions, they can be transformed into consistent (though
not efficient) estimators of $\rho$ (see, for example, \citet{Croux2010}).
Despite their robustness and consistency, estimators based on rank
correlation measures typically have sampling distributions that remain
sensitive to the underlying data generating process. This sensitivity
poses a substantial challenge for inference on correlations in applied
empirical analysis, where the true distribution of the data is typically
unknown.

A related strand of the robustness literature builds inference on
the directions of the observations alone. Since the direction, or
spatial sign, of an elliptical random vector has a distribution that
does not depend on the radial component, sign-based procedures are
naturally insensitive to extreme magnitudes and heavy tails. Prominent
examples include the distribution-free M-estimator of scatter of \citet{tyler1987distribution},
sign and rank covariance matrices (\citet{visuri2000sign}, \citet{oja2010multivariate}),
the spatial sign correlation of \citet{durre2015spatial}, and semiparametrically
efficient rank-based inference for shape parameters (\citet{hallin2006semiparametrically1}).
The similarity measure studied in this paper belongs to this family,
as it depends on an observation only through its direction. Its distinctive
feature, however, is that, on the Fisher scale, it admits an \emph{exact}
finite-sample distribution that is invariant over the entire elliptical
class. Consequently, it supports exact robust inference for correlations
in small samples, rather than only asymptotically distribution-free
point estimation.

In this paper, we develop a class of estimators of statistical association
between two random variables that are consistent for a specific functional
of the covariance, or scatter, matrix of the underlying variables.
The estimators build on a measure of similarity between two variables,
inspired by the works of \citet{Thorndike_1905} and \citet{Fisher_1919},
that accounts for agreement in both sign and magnitude. Under elliptically
distributed data with homogeneous scales, the similarity variable
follows, exactly and for every member of the elliptical family, a
hyperbolic secant distribution centered at the Fisher transformation
of $\rho$. This result motivates two estimators that operate directly
on the Fisher scale. The first is a sample-average similarity estimator,
whose finite-sample distribution is available in semi-explicit form
via the characteristic function and is invariant over the class of
elliptical distributions with arbitrary kurtosis; this enables not
only robust estimation but also exact inference for correlations,
such as interval estimation with a prescribed coverage probability,
even in very small samples. The second is a maximum-likelihood similarity
estimator, which solves a strictly concave likelihood problem and
reduces the asymptotic variance from $\frac{\pi^{2}}{4}\approx2.467$
to $2$, the inverse Fisher information of the hyperbolic secant family,
while preserving the same distributional invariance. Naturally, the
efficiency of both estimators is lower than that of the sample correlation
estimator, reflecting the price paid for intrinsic robustness; the
resulting confidence intervals are wider on average, but remain valid
under outliers and extremely heavy-tailed data.

In the more realistic scenario where the scales are not homogeneous,
the similarity estimators remain meaningful: they consistently estimate
the Fisher transformation of the coefficient of resemblance, an alternative
measure of association that accounts jointly for the correlation and
for the similarity of scales, and that provides a lower bound on the
magnitude of $\rho$. To recover the correlation coefficient itself,
we develop a two-stage procedure in which the variables are first
re-scaled toward homogeneous dispersion using the median of the observed
log-ratios -- a robust, purely angular statistic available in closed
form -- after which the (maximum-likelihood) similarity estimator
is applied to the re-scaled data. We prove that the resulting scale-equalized
estimator is strongly consistent and asymptotically normal, with the
same limit variance as in the homogeneous case, for every elliptical
generator.

We propose a natural generalization of the similarity measure to the
multivariate setting involving an arbitrary number of variables. The
resulting generalized similarity measure captures the relative variation
of the observed data along a similarity direction, defined by the
vector of ones, and can be interpreted as a measure of joint similarity
across multiple variables. This measure naturally inherits robustness
against outliers and nests the bivariate similarity measure as a special
case. For elliptically distributed data with homogeneous dispersion
and identical pairwise correlations, the multivariate similarity estimator
becomes a consistent estimator of the (transformed) equicorrelation
parameter, with an exact and invariant sampling distribution, thus
permitting a correlation-based interpretation analogous to the bivariate
case.

Estimating financial correlations from high-frequency data represents
a particularly promising area for applications. While intraday data
provide a rich source of observations, offering the potential for
efficient estimation and accurate inference, the analysis of such
data is accompanied by multiple econometric challenges. For example,
large instantaneous price movements, or jumps, which are prevalent
in financial markets, generate extreme (outlying) return observations
that often induce a downward bias in estimated correlations. In addition,
at sufficiently high frequencies, the estimation of realized covariances
and correlations may be affected by various forms of market microstructure
noise (see \citet{Hansen_Lunde_2006}, \citet{BandiRussellRES08}),
data asynchronicity and the Epps effect (see \citet{Ren`o2003}),
as well as other adverse artifacts.

The literature on covariance and correlation estimation using high-frequency
data proposes a wide range of estimators with varying degrees of robustness
to the aforementioned challenges. In practice, realized correlations
are most commonly obtained from multivariate volatility estimators
by rescaling estimated covariances to the corresponding linear correlations
(see \citet{BNS:2004}, \citet{AitSahaliaFanXiu2010}, \citet{BNHLS:2011},
\citet{HansenHorelLundeArchakov}, among others). While many of these
methods are explicitly designed to handle various microstructural
effects, they nevertheless remain vulnerable to price jumps and other
outliers. A common remedy is the application of truncation or thresholding
techniques to filter out extreme returns (see, for example, \citet{Mancini2001},
\citet{AndersenDobrevSchaumburg2012}). An alternative approach is
to employ robust correlation measures, such as the Kendall and Quadrant
correlation estimators, along with their modifications and adaptations
for high-frequency data (see \citet{VanderElst_Veredas_2015}, \citet{Hansen_Luo_2023},
among others).

The intrinsic robustness of the similarity estimator to extreme observations
and heavy-tailed data makes it a promising tool for correlation estimation
in high-frequency data settings. A key feature that distinguishes
the similarity estimator from existing alternatives is its ability
to deliver robust interval estimates for correlations, enabled by
the availability of a robust and invariant sampling distribution.
We provide an empirical illustration by constructing robust confidence
intervals for daily correlations of $21$ liquid U.S. stocks from
seven sectors over the sixteen-year period between 2005 and 2020.
The analysis covers more than $800,000$ pairwise daily correlations,
each estimated at thirteen sampling frequencies ranging from $2$
seconds to $30$ minutes. For moderate frequencies ($\Delta\ge30$
sec), the estimated robust intervals show remarkably strong agreement
with robust correlation estimates based on the Kendall coefficient,
with the proportion of matches approaching one, supporting the reliable
performance of the similarity estimator in high-frequency applications.
At the same time, the corresponding proportion for traditional realized
correlations is substantially lower and decreases steadily as the
sampling frequency rises, consistent with the sensitivity of the sample
correlation to price jumps and other outliers.

We note that, in our empirical analysis, the estimator is applied
while ignoring several recognized features of intraday data, such
as time-varying volatilities and correlations, which are inconsistent
with the assumption of independence of sample observations. Accordingly,
we leave a careful adaptation of the proposed estimator to high-frequency
settings for future research.

Another area for applications is the modeling of correlation dynamics
for vectors of asset returns. The traditional approach builds on the
extensive literature on multivariate GARCH or stochastic volatility
models, in which the conditional correlation process is assumed, either
directly or indirectly, to evolve over time (see the respective surveys
in \citet{Bauwens_Laurent_Rombouts_2006} and \citet{Asai_McAleer_Yu_2006}).
The models are typically applied to relatively low-frequency data,
such as daily or weekly asset returns. Key challenges of this class
of models include preserving positive definiteness of the conditional
correlation structures, as well as the rapidly increasing computational
complexity as the number of assets grows.

We propose a new class of multivariate GARCH models that employ the
similarity measure to model correlation dynamics. Our approach follows
the Dynamic Conditional Correlation (DCC) framework introduced in
\citet{Engle_2002}, in which the conditional correlation process
is modeled in isolation from the conditional volatilities. We develop
a parsimonious model specification in which an arbitrary number of
assets can be accommodated without an increase in computational complexity.
To this end, we impose the equicorrelation assumption, similar to
\citet{Engle_Kelly_2012}, and specify the dynamics of the conditional
correlation parameter under the matrix logarithmic transformation.
This formulation eliminates the need for additional constraints to
guarantee positive definiteness, in contrast to the classical DCC
framework; in the bivariate case, for example, the correlation dynamics
is specified directly on the Fisher scale, ensuring positive definiteness
by construction. The robust similarity measure, computed on standardized
returns, is then used as an observation-driven update for the dynamic
correlation parameter, and the specification accommodates an asymmetric
(leverage-type) response of correlations to the prevailing direction
of market returns. An important feature of the proposed model is the
intrinsic robustness of the estimated correlation process to extreme
observations and fat-tailed distributions, which are common in financial
returns. Because the similarity signal depends only on the direction
of the standardized returns, the proposed correlation dynamics remains
well defined across the entire elliptical family, including infinite-variance
members for which classical DCC-type recursions have no population
counterpart.

We apply the proposed model to daily returns on $49$ U.S. industry
portfolios over a sample period spanning almost six decades ($1969$
to $2026$). The estimated conditional equicorrelation index reveals
pronounced low-frequency cyclical dynamics, with the common correlation
level ranging from about $0.2$ to $0.8$ and major global downturns
marking the turning points of long-term correlation cycles. Compared
with the standard DECO specification, the new robust index responds
more modestly to short-lived episodes of co-directional extreme returns,
such as the market crashes of $1987$ and $1997$ and the outbreak
of COVID-19, pointing to distinctive informational content in the
robust correlation index. A thorough assessment of its statistical
and economic performance is left for future research.

\section{\label{sec:Measure-of-Similarity}A Measure of Similarity for Random
Variables}

Let $x_{1}$ and $x_{2}$ be two real-valued random variables. Assume
additionally that $x_{1}$ and $x_{2}$ have zero mean (or zero location
parameters, $\mu_{1}=\mu_{2}=0$), and the norm of the random vector
$x=(x_{1},x_{2})^{\prime}$ is positive with probability one. We consider
the following variable,
\begin{equation}
r=\frac{2x_{1}x_{2}}{x_{1}^{2}+x_{2}^{2}},\label{eq:res}
\end{equation}
which can be interpreted as a measure of statistical similarity between
$x_{1}$ and $x_{2}$. For empirical analysis, this quantity was introduced
in \citet{Thorndike_1905} to measure statistical association between
a pair of twins with respect to a set of considered characteristics,
and was originally named \textit{resemblance}. In that study, this
term was used as a synonym for the \textit{coefficient of correlation}. 

Indeed, multiple aspects allow one to consider the Thorndike's resemblance
variable, $r$, as a measure of correlation. The quantity is confined
within the fixed interval, $r\in[-1,1]$, and the value of $r$ is
higher when $x_{1}$ and $x_{2}$ are more similar in magnitude, while
sharing the same sign. The magnitude of $r$ is larger when magnitudes
of $x_{1}$ and $x_{2}$ are more similar to each other, while the
sign of $r$ is positive (negative) if $x_{1}$ and $x_{2}$ have
the same (opposite) directions.

Variable $r$ has a particularly simple form when the random vector
is represented in polar coordinates ($x_{1}=s\cos\theta$ and $x_{2}=s\sin\theta$).
In this case, $r=\sin2\theta$, so the resemblance measure does not
depend on the vector length, $s$. Intuitively, it implies that $r$
depends only on the angle between the observed vector and coordinate
axes and ignores the magnitude of observations. This points to intrinsic
insensitivity of the resemblance measure, $r$, to the presence of
outliers in the data. 

If we additionally assume that vector $x=(x_{1},x_{2})^{\prime}$
has finite second moments, the Pearson correlation between $x_{1}$
and $x_{2}$ is defined as
\begin{equation}
\rho=\frac{\mathbb{E}(x_{1}x_{2})}{\sqrt{\mathbb{E}(x_{1}^{2})\mathbb{E}(x_{2})^{2}}}.\label{eq: Pearson_corr}
\end{equation}
The Pearson correlation is the most popular measure of statistical
association that arises naturally in multiple econometric frameworks,
from a simple linear regression to complex statistical learning algorithms.
The Pearson correlation is often considered as the default benchmark
for measuring dependence in empirical analysis. Although $\rho$ is
able to capture the precise association between random variables only
when the true relationship is linear, it can still provide a reasonable
approximation when the underlying dependence is monotonic and not
heavily non-linear.

While the correlation coefficient $\rho$ is bounded between $-1$
and $1$ by construction, it is sometimes convenient to work with
an unconstrained correlation measure and this can be achieved by means
of a suitable transformation. The most prominent example is the Fisher
transformation defined as,
\[
\phi_{\rho}=\frac{1}{2}\log\Bigl(\frac{1+\rho}{1-\rho}\Bigl),
\]
for $\rho\ensuremath{\in(-1,1)}$. The Fisher transformation represents
a strictly monotone transformation of $\rho$ onto the set of real
numbers, such that $\phi_{\rho}\in\mathbb{R}$. The transformation
was proposed by Ronald A. Fisher in a series of seminal papers (see
\citet{Fisher_1915}, \citet{Fisher_1921}), where he also demonstrated
that it improves the distributional properties of the sample correlation
coefficient.

In what follows, we will refer to the Thorndike's resemblance measure,
$r$, under the Fisher transformation, 
\begin{equation}
\phi_{r}=\frac{1}{2}\log\Bigl(\frac{1+r}{1-r}\Bigl)=\frac{1}{2}\log\frac{(x_{1}+x_{2})^{2}}{(x_{1}-x_{2})^{2}},\label{eq:phi_r}
\end{equation}
as to the \textit{measure of similarity}, or the similarity variable.
In contrast to $r$, the similarity measure $\phi_{r}$ has an unrestricted
range, $\phi_{r}\in\mathbb{R}$, and its magnitude increases as the
values of $x_{1}$ and $x_{2}$ become more similar. The sign of $\phi_{r}$
is positive (negative) when the two observations have the same (opposite)
signs. Figure \ref{fig:contour_gamma} illustrates the magnitudes
of $\phi_{r}$ as a function of $x_{1}$ and $x_{2}$, where warmer
(red) colors indicate higher values of $\phi_{r}$ and cooler (blue)
colors indicate lower values. We also note that, similar to $r$,
$\phi_{r}$ depends only on the angular coordinate of the vector $x$,
which underlies its intrinsic robustness to observations with extreme
magnitudes.

\begin{figure}[h]
\begin{centering}
\includegraphics[scale=0.6]{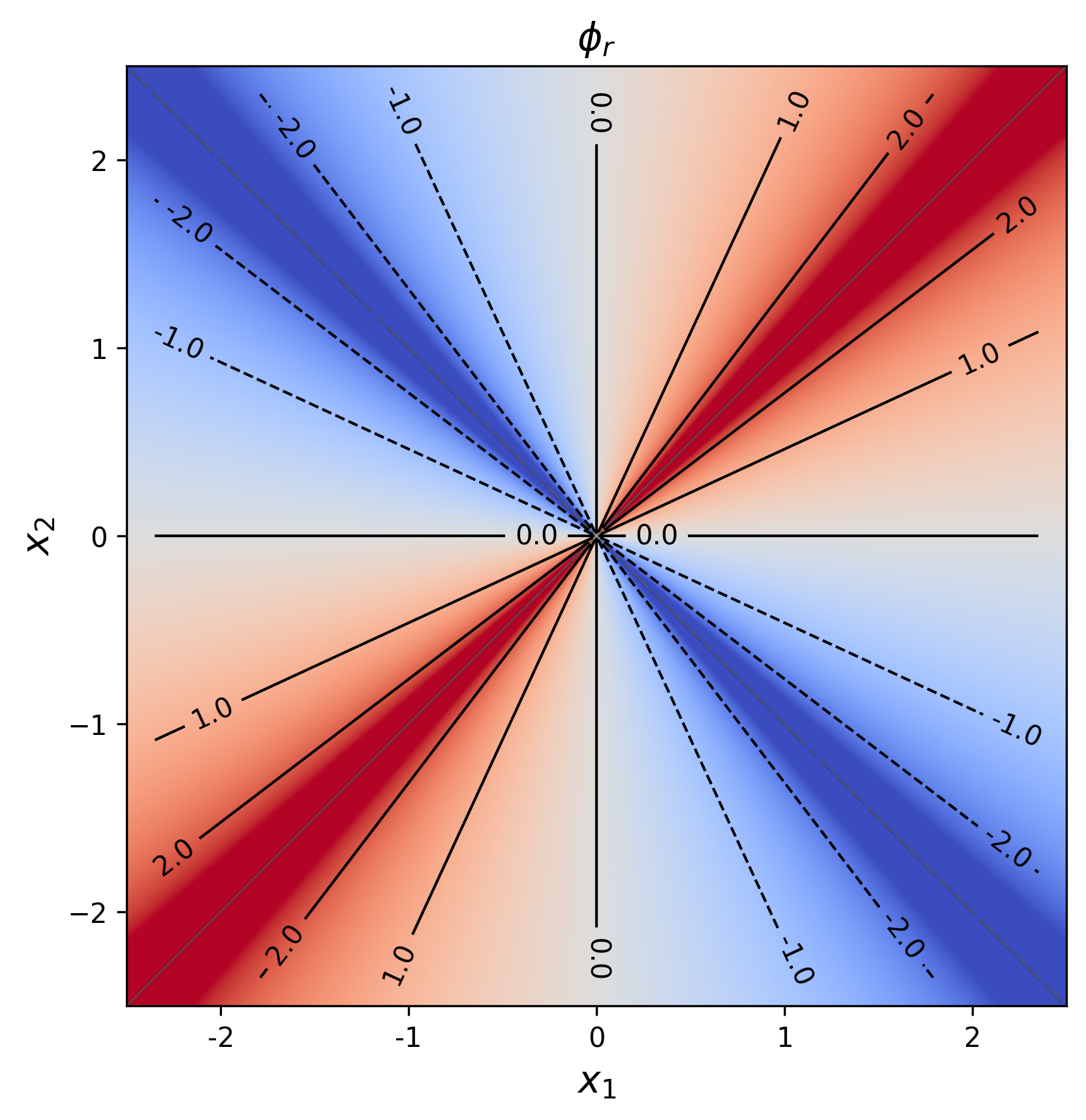}
\par\end{centering}
\caption{\footnotesize\label{fig:contour_gamma} A heatmap plot for $\phi_{r}$
as a function of $x_{1}$ and $x_{2}$. Areas with the red color indicate
higher values of $\phi_{r}$, while areas with the blue color indicate
lower (negative) values of $\phi_{r}$. The white lines corresponding
to $x_{1}=x_{2}$ ($r=1$) and $x_{1}=-x_{2}$ ($r=-1$) represent
the loci where the function is undefined.}
\end{figure}

An important special case arises when the random vector $x$ follows
a bivariate elliptical distribution, for which the dependence structure
is inherently linear. For this case, let the scatter matrix of $x$
be a positive definite matrix given by
\begin{equation}
\Sigma=\left(\begin{array}{cc}
\sigma_{1}^{2} & \sigma_{12}\\
\sigma_{12} & \sigma_{2}^{2}
\end{array}\right),\label{eq:Sigma_2x2_hetero}
\end{equation}
where the marginal scale parameters of $x_{1}$ and $x_{2}$ are denoted
by $\sigma_{1}^{2}$ and $\sigma_{2}^{2}$, respectively. The linear
correlation parameter is defined by $\rho=\frac{\sigma_{12}}{\sqrt{\sigma_{1}^{2}\sigma_{2}^{2}}}\in(-1,1)$
and can be considered as a natural generalization of the Pearson correlation
for the elliptical family. Note that $\rho$ is well defined even
if second moments of $x$ do not exist (e.g., as for the multivariate
Cauchy distribution), but otherwise is equivalent to the classical
Pearson correlation in \eqref{eq: Pearson_corr}.\footnote{For zero location elliptical vectors, the linear correlation $\rho$
can also be equivalently defined via the transformed quadrant probabilities
(see \citet{Sheppard_1899}, \citet{Greiner_1909}, etc.)}

Assume additionally that the scales are homogeneous, such that $\sigma_{1}=\sigma_{2}=\sigma$.
Under this assumption, $\phi_{r}$ and $\phi_{\rho}$ are elegantly
connected. This result was originally formulated in \citet{Fisher_1919}
for the Gaussian case, and below we provide an extension of the original
result to the entire class of elliptical distributions.
\begin{thm}[R. Fisher, 1919]
\label{prop:T1} Assume that $x=(x_{1},x_{2})^{\prime}$ is a bivariate
random vector which follows some elliptical distribution with zero
location and positive-definite scatter matrix $\Sigma$ with homogeneous
scale parameters, and let $\rho$ denote the linear correlation between
$x_{1}$ and $x_{2}$. Denote the measure of resemblance by $r=\frac{2x_{1}x_{2}}{x_{1}^{2}+x_{2}^{2}}$,
and the corresponding Fisher transformation by $\phi_{r}=\frac{1}{2}\log\Bigl(\frac{1+r}{1-r}\Bigl)$.
Then, $\phi_{r}$ is a random variable with the probability density
function
\[
f(\phi_{r})=\frac{1}{\pi}\text{sech}(\phi_{r}-\phi_{\rho}),
\]
where $\phi_{\rho}=\frac{1}{2}\log\Bigl(\frac{1+\rho}{1-\rho}\Bigl)$
is the Fisher transformation of $\rho$.
\end{thm}
Theorem \ref{prop:T1} implies that $\phi_{r}$ is symmetrically distributed,
according to the hyperbolic secant distribution, around the Fisher
transformation of $\rho$. The distribution of $\phi_{r}-\phi_{\rho}$
depends neither on the scale parameter, $\sigma^{2}$, nor on the
underlying correlation coefficient, $\rho$, nor on the elliptical
generator. Therefore, under the assumptions of the proposition, the
similarity variable $\phi_{r}$ provides an unbiased and robust signal
of the latent correlation level (on the Fisher scale) with stable
sampling properties, and thus emerges as an attractive statistical
tool for correlation estimation.

\section{\label{sec:Realized-Similarity-Estimator}The Similarity Estimator }

Let a random sample be given, $\{x_{t}\}_{t=1}^{T}$, where $x_{t}=(x_{1,t},x_{2,t})^{\prime}$
are independent observations from some bivariate elliptical distribution
with zero location and positive definite scatter matrix as in \eqref{eq:Sigma_2x2_hetero}
with the linear correlation coefficient, $\rho\in(-1,1)$. Probably
the most popular estimator of the correlation is the sample correlation
estimator given by
\[
\hat{\rho}=\frac{\sum_{t=1}^{T}x_{1,t}x_{2,t}}{\sqrt{\Bigl(\sum_{t=1}^{T}x_{1,t}^{2}\Bigl)\Bigl(\sum_{t=1}^{T}x_{2,t}^{2}\Bigl)}},
\]
where we impose $\mathbb{E}(x)=0$, if the expectation exists. Under
the mild assumption of finite first and second moments, the sample
correlation, $\hat{\rho}$, is a consistent estimator of the population
linear correlation, and is asymptotically normal with $\sqrt{T}(\hat{\rho}-\rho)\overset{d}{\rightarrow}\mathcal{N}\Bigl(0,V_{\rho}\Bigl)$.
For relatively small $T$, however, the sampling properties of $\hat{\rho}$
are often poorly approximated by the asymptotic results, especially
in the presence of sufficiently heavy-tailed data. The inference is
additionally complicated by the fact that the asymptotic variance
$V_{\rho}$ generally depends on the unknown value of $\rho$. For
example, for the Gaussian case, $V_{\rho}=(1-\rho^{2})^{2}$; see
\citet{Fisher_1915}.

The Fisher transformation is particularly useful in improving sampling
properties of $\hat{\rho}$. We denote the Fisher transformation of
$\hat{\rho}$ by $\phi_{\hat{\rho}}$. Thus, for normally distributed
observations, the asymptotic distribution of $\phi_{\hat{\rho}}$
is $\sqrt{T}(\phi_{\hat{\rho}}-\phi_{\rho})\overset{d}{\rightarrow}\mathcal{N}\Bigl(0,1\Bigl)$,
and the asymptotic variance is independent of $\rho$. More importantly,
the Fisher transformation offers multiple advantages in finite samples.
In particular, when the data is close to be normally distributed,
it provides the variance stabilization for the sampling distribution
of $\phi_{\hat{\rho}}$ with making it symmetric and nearly Gaussian
even for very small samples (see \citet{Fisher_1921}, \citet{Hotelling_1953},
etc.). These properties often motivate to analyze the sample correlation
coefficient in the Fisher scale when conducting inference for correlations. 

The critical drawback of the sample correlation estimator, that may
compromise both estimation and inference, is its notorious sensitivity
to outliers and, more generally, to extreme observations. A popular
example of a robust correlation measure is the Kendall rank correlation
coefficient, or Kendall\textquoteright s tau coefficient (\citet{esscher1924method},
\citet{Kendall_1938}), given by
\[
\hat{\tau}=\frac{2}{T(T-1)}\sum_{i<j}\text{sign}(x_{1,i}-x_{1,j})\text{sign}(x_{2,i}-x_{2,j}).
\]
The measure captures the strength of monotonic dependence between
two variables by quantifying the relative frequency of sign-concordant/discordant
pairs of observations. For elliptical distributions, it can be transformed
to match the Pearson correlation via the Greiner's equality which
links the correlation coefficient with the quadrant probabilities,
$\rho=\sin\Bigl[\frac{\pi}{2}\mathbb{E}(\hat{\tau})\Bigl]$, see \citet{Greiner_1909}.
A similar alternative is the class of quadrant estimators of correlation
which are based on the sample proportion of sign-concordant observations
(see \citet{Sheppard_1899}, \citet{Kendall_1949}, \citet{Blomqvist_1950},
\citet{Hansen_Luo_2023}, etc.)

The results in Proposition \ref{prop:T1} motivate an alternative
estimator of statistical association that is based on the Fisher transformed
similarity, $\phi_{r}$. For each observation $x_{t}$, $t=1,...,T$,
we can construct the corresponding (local) empirical measure of similarity
$\phi_{r,t}$, as given in \eqref{eq:phi_r}. In case the scale parameters
of $x_{t}$ are identical, it follows from Theorem \ref{prop:T1}
that $\mathbb{E}\phi_{r,t}=\phi_{\rho}$. Therefore, $\phi_{r,t}$
is an unbiased and robust signal of the correlation coefficient, on
the Fisher scale, and this naturally motivates suggesting the following
moment-based estimator,
\begin{equation}
\hat{\gamma}=\frac{1}{T}\sum_{t=1}^{T}\phi_{r,t}=\frac{1}{2T}\sum_{t=1}^{T}\log\frac{(x_{1,t}+x_{2,t})^{2}}{(x_{1,t}-x_{2,t})^{2}}.\label{eq:realized_sim}
\end{equation}
In what follows, we will refer to $\hat{\gamma}$ as to the similarity
estimator. In the elliptical case with homogeneous variances, the
similarity estimator is i) a consistent and unbiased estimator of
the linear correlation coefficient (in the Fisher scale), ii) has
robust mean and robust sampling distribution that does not depend
on $\rho$, iii) the sampling distribution is available via the characteristic
function, for any $T$, allowing for exact inference on correlations.
In a more general scenario, when scales of $x_{1}$ and $x_{2}$ may
differ, the estimator $\hat{\gamma}$ can be treated as a standalone,
robust measure of statistical association between the two random variables,
and it retains an additional interpretation as a lower bound on the
correlation coefficient $\rho$.

\subsection{\label{subsec:Variance-Homogeneity:-Robust}Robust Estimation and
Inference for Correlations under Scale Homogeneity}

By Theorem \ref{prop:T1}, under scale homogeneity, $\phi_{r,t}$
are i.i.d. with the hyperbolic secant density centered at $\phi_{\rho}$,
which is symmetric around the center and has finite moments of all
orders. Hence, $\hat{\gamma}\overset{a.s.}{\rightarrow}\mathbb{E}\phi_{r,t}=\phi_{\rho}$
and 
\begin{equation}
\sqrt{T}(\hat{\gamma}-\phi_{\rho})\overset{d}{\rightarrow}\mathcal{N}\Bigl(0,\frac{\pi^{2}}{4}\Bigl),\label{eq:gamma_clt}
\end{equation}
as $\mathbb{E}\phi_{r,t}=\phi_{\rho}$ and $\mathbb{V}(\phi_{r,t})=\frac{\pi^{2}}{4}$.

The asymptotic distribution of $\hat{\gamma}$ does not depend on
the actual correlation coefficient $\rho$. In contrast to $\phi_{\hat{\rho}}$,
which represents the transformation of the sample correlation estimator
$\hat{\rho}$, the similarity estimator, $\hat{\gamma}$, directly
targets $\phi_{\rho}$ by estimating the correlation coefficient on
the Fisher scale. The asymptotic variance of $\hat{\gamma}$ is $\frac{\pi^{2}}{4}$,
and this value is, in general, larger than the asymptotic variance
of $\phi_{\hat{\rho}}$ in the unconstrained scenario. The efficiency
reduction is not surprising since $\hat{\gamma}$ incorporates only
information about the relative magnitudes of $(x_{1,t},x_{2,t})$
and their signs, but ignores information about the total magnitude
of $x_{t}$. A lower efficiency of $\hat{\gamma}$ is nonetheless
compensated by its robustness to outliers and stability of the sampling
distribution.

A remarkable feature of the similarity estimator in the considered
scenario is that not only the asymptotic distribution, but also the
finite-sample distribution of $\hat{\gamma}-\phi_{\rho}$ is invariant
over the entire class of elliptical distributions of $x_{t}$. Furthermore,
the characteristic function of $\phi_{r,t}$ is available and allows
to recover the exact sampling distribution of the similarity estimator
for any finite $T$. Specifically, the distribution $f(\phi_{r,t})$
provided in Proposition \ref{prop:T1} implies that the characteristic
function for $\phi_{r}-\phi_{\rho}$ is given by $\varphi_{\phi_{r}}(u)=\text{sech}\bigl(\frac{\pi}{2}u\bigl)$.
Denote the standardized similarity estimator by $z_{\hat{\gamma}}=\frac{2\sqrt{T}}{\pi}(\hat{\gamma}-\phi_{\rho})$.
Then, the characteristic function of $z_{\hat{\gamma}}$ can be written
as 
\begin{equation}
\varphi_{z}(u)=\Pi_{t=1}^{T}\varphi_{\phi_{r,t}}\Bigl(\frac{2u}{\pi\sqrt{T}}\Bigl)=\Biggl[\text{sech}\Bigl(\frac{u}{\sqrt{T}}\Bigl)\Biggl]^{T},\label{eq:cf}
\end{equation}
for $u\in\mathbb{R}$. Therefore, the sampling distribution of $\hat{\gamma}$
becomes available in semi-explicit form (via the characteristic function)
for any finite $T$, and this allows to conduct exact inference for
the estimated correlation parameter. The sampling distribution of
$z_{\hat{\gamma}}$ is symmetric with a positive excess kurtosis for
any sample size. It converges to the standard normal distribution
very quickly as $T$ increases. This is illustrated in Figure \ref{fig:pdf1}
and in Table \ref{tab:quantiles_table}, where the quantiles of the
finite-sample distribution are reported for a range of $T$. 

\begin{figure}[h]
\begin{centering}
\includegraphics[scale=0.8]{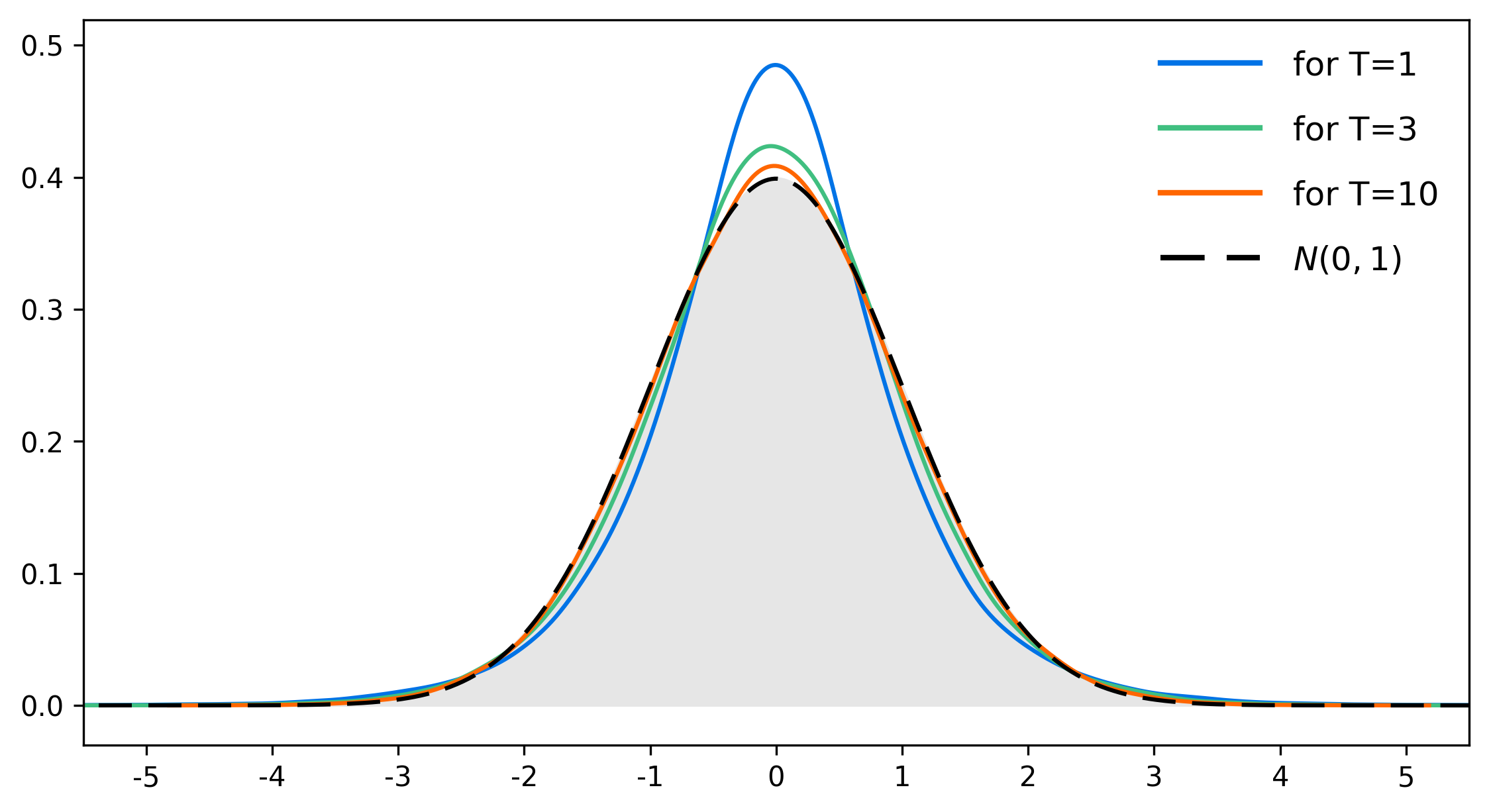}
\par\end{centering}
\caption{\footnotesize\label{fig:pdf1} The probability density functions
of $z_{\hat{\gamma}}=\frac{\sqrt{T}(\hat{\gamma}-\phi_{\rho})}{\pi/2}$
for $T=1,3,10$ (colored lines) and the standard normal probability
density (dashed line) which is the distribution limit of $z_{\hat{\gamma}}$
for $T\rightarrow\infty$.}
\end{figure}

The similarity estimator $\hat{\gamma}$ can be used as a robust and
consistent estimator of the linear correlation coefficient once the
inverse Fisher transformation is applied, $\phi^{-1}(\hat{\gamma})$,
where
\[
\phi^{-1}(\gamma)=\frac{e^{2\gamma}-1}{e^{2\gamma}+1}=\tanh(\gamma),
\]
which is a hyperbolic tangent function. We note that although $\hat{\gamma}$
is an unbiased estimator of $\phi_{\rho}$, the inverse transformation,
$\phi^{-1}(\hat{\gamma})$, is not an unbiased estimator of $\rho$
because $\mathbb{E}\phi^{-1}(\hat{\gamma})\neq\rho$, in general,
for $T>1$. However, with an exact finite-sample distribution of $\hat{\gamma}$
and due to monotonicity of the Fisher transformation, $\hat{\gamma}$
can be used for \textit{interval estimation} of $\rho$ by providing
exact confidence intervals for any sample size $T$.

Under the scale homogeneity assumption, another interesting feature
of the similarity estimator is related to the matrix logarithm transformation.
Assuming $x_{t}$ has a positive definite scatter matrix $\Sigma$
with the common scale parameter $\sigma^{2}$ and correlation $\rho$
, $\lambda_{+}=\sigma^{2}(1+\rho)$ and $\lambda_{-}=\sigma^{2}(1-\rho)$
are the two eigenvalues of $\Sigma$ with the corresponding eigenvectors
$q_{+}=\frac{1}{\sqrt{2}}\iota_{2}$ and $q_{-}=\frac{1}{\sqrt{2}}\iota_{2}^{\perp}$,
where $\iota_{2}=(1,1)^{\prime}$ and $\iota_{2}^{\perp}=(1,-1)^{\prime}$.
For $n=2$, the matrix logarithm transformation of $\Sigma$ has an
explicit analytic expression and reads 
\[
\log\Sigma=\left(\begin{array}{cc}
\frac{1}{2}\log(\lambda_{+}\cdot\lambda_{-}) & \mathinner{\color{purple}\frac{1}{2}}{\color{purple}\log\Bigl(\frac{\lambda_{+}}{\lambda_{-}}}\mathopen{\color{purple}\Bigl)}\\
\mathinner{\color{purple}\frac{1}{2}}{\color{purple}\log\Bigl(\frac{\lambda_{+}}{\lambda_{-}}}\mathopen{\color{purple}\Bigl)} & \frac{1}{2}\log(\lambda_{+}\cdot\lambda_{-})
\end{array}\right),
\]
where the off-diagonal entry, representing the (half) log-condition
number of $\Sigma$, coincides with the Fisher transformation of the
correlation coefficient, $\phi_{\rho}=\frac{1}{2}\log\Bigl(\frac{\lambda_{+}}{\lambda_{-}}\Bigl)$.
Therefore, under the considered assumptions, $\hat{\gamma}=\frac{1}{2T}\sum_{t=1}^{T}\log\frac{(q_{+}^{\prime}x_{t})^{2}}{(q_{-}^{\prime}x_{t})^{2}}$
consistently estimates the off-diagonal element of $\log\Sigma$.
This suggests a promising direction for using $\hat{\gamma}$ to estimate
correlation matrices directly under the matrix logarithm transformation
which offers many convenient properties for correlation analysis (see
\citet{Archakov_Hansen_2021}). We further explore this idea in Section
\ref{subsec:Equicorrelation-Scenario}.

\subsection{\label{subsec:Maximum-Likelihood-Robust}Maximum-Likelihood Robust
Similarity Estimator }

The moment-based estimator $\hat{\gamma}$ in \eqref{eq:realized_sim}
does not efficiently employ the information about the distribution
of $\phi_{r}$ provided in Theorem \ref{prop:T1}. In contrast, the
maximum likelihood (ML) estimation exploits this information in its
entirety and provides an asymptotically more efficient estimator of
$\phi_{\rho}$. For a sample of $\phi_{r,t}$, $t=1,...,T$, the log-likelihood
function is given by
\[
\ell(\gamma)=-T\log\pi+\sum_{t=1}^{T}\log\text{sech}(\phi_{r,t}-\gamma),
\]
and the corresponding score and Hessian are
\[
\nabla_{\gamma}\ell=\sum_{t=1}^{T}\text{tanh}(\phi_{r,t}-\gamma),\qquad\text{and}\qquad\nabla_{\gamma}^{2}\ell=-\sum_{t=1}^{T}\text{sech}^{2}(\phi_{r,t}-\gamma),
\]
where $\nabla_{\gamma}^{2}\ell<0$. The natural candidate estimator
is a solution of the necessary condition, $\nabla_{\gamma}\ell=0$.
The following proposition formulates the maximum-likelihood robust
similarity estimator.
\begin{prop}
\label{prop:P1} For a random sample of bivariate elliptical vectors,
$x_{t}=(x_{1,t},x_{2,t})^{\prime}$, for $t=1,...,T$, with zero location
and positive-definite scatter matrix $\Sigma$ with homogeneous scales
and correlation parameter $\rho\in(-1,1)$, the equation
\[
\sum_{t=1}^{T}\text{tanh}(\phi_{r,t}-\gamma)=0,
\]
with $\phi_{r,t}$ given in \eqref{eq:phi_r}, admits a unique solution,
$\hat{\gamma}_{ML}$, almost surely. The maximum-likelihood robust
similarity estimator $\hat{\gamma}_{ML}$ is strongly consistent for
$\phi_{\rho}=\frac{1}{2}\log\Bigl(\frac{1+\rho}{1-\rho}\Bigl)$ and
has the asymptotic distribution 
\[
\sqrt{T}(\hat{\gamma}_{ML}-\phi_{\rho})\overset{d}{\rightarrow}\mathcal{N}\Bigl(0,\mathcal{I}^{-1}(\phi_{\rho})\Bigl),
\]
where $\mathcal{I}(\phi_{\rho})=\frac{1}{2}$ is the Fisher Information
of the hyperbolic secant family.
\end{prop}
Although $\hat{\gamma}_{ML}$ does not admit a closed-form expression,
the numeric estimation is fast, especially when it is accelerated
with the use of analytical derivatives. The gain in asymptotic efficiency
of $\hat{\gamma}_{ML}$ relative to $\hat{\gamma}$ is nonetheless
appreciable: the asymptotic variance falls from $\frac{\pi^{2}}{4}\approx2.467$
to $2$.

\begin{figure}[h]
\begin{centering}
\includegraphics[scale=0.8]{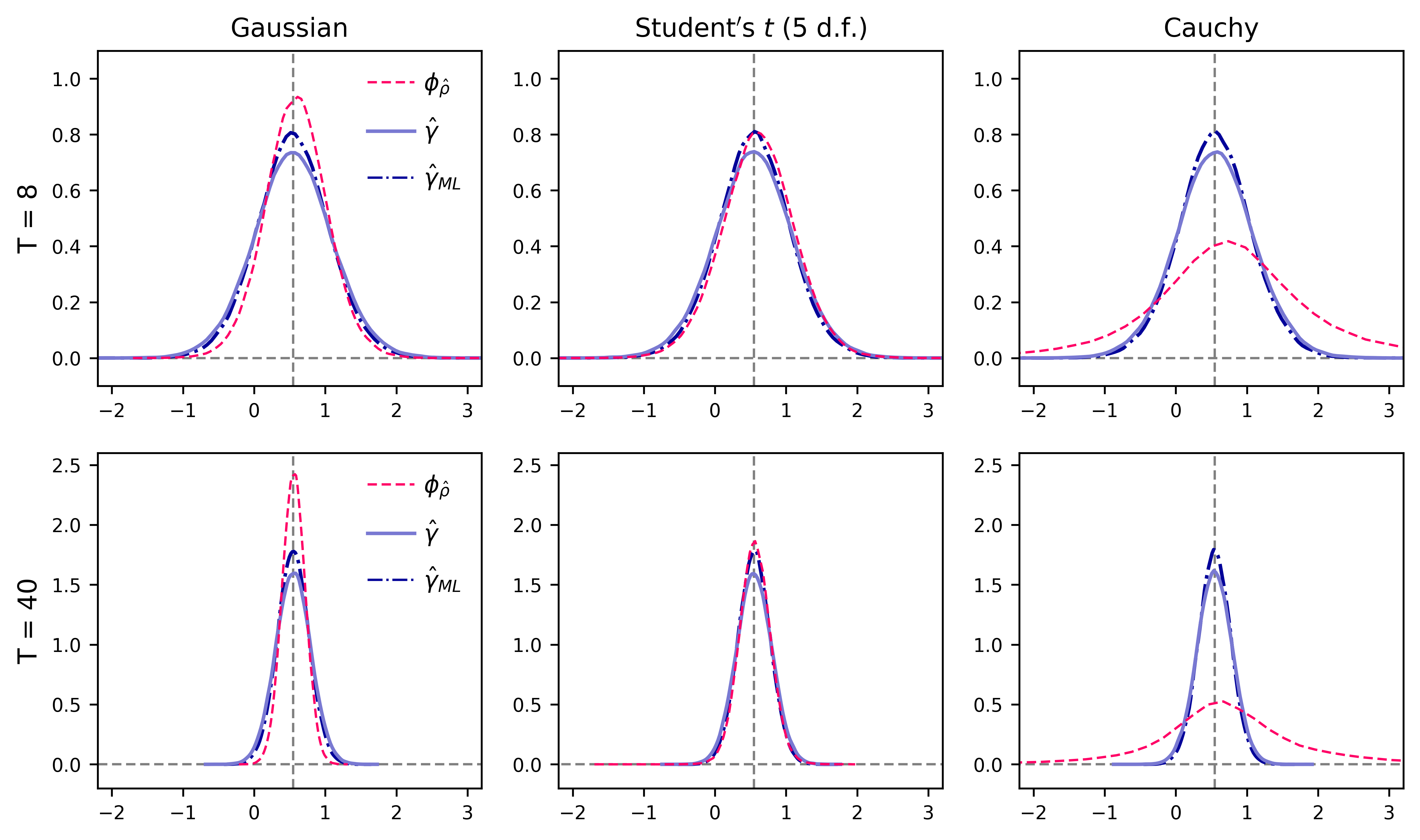}
\par\end{centering}
\caption{\footnotesize\label{fig:finite_sample_densities} Finite-sample distributions
of $\phi_{\hat{\rho}}$ (red dashed lines), $\hat{\gamma}$ (blue
solid lines) and $\hat{\gamma}_{ML}$ (dark dash-dotted lines) obtained
on 100,000 simulated samples of sizes $T=8$ (top plots) and $T=40$
(bottom plots). Data vectors $x_{t}$ were simulated out of the three
selected distributions -- normal, $t$-distribution with 5 degrees
of freedom, and Cauchy -- with the true correlation parameter $\rho=0.5$.}
\end{figure}

Figure \ref{fig:finite_sample_densities} provides small sample distributions
of $\hat{\gamma}$, $\hat{\gamma}_{ML}$ and $\phi_{\hat{\rho}}$
resulted from the simulation analysis. In this illustration we consider
two selected sample sizes ($T=8$ and $T=40$) and three selected
bivariate distributions of vector $x_{t}$: normal, $t$-distribution
with 5 degrees of freedom, and Cauchy. The figure shows an apparent
instability of the sampling density of $\phi_{\hat{\rho}}$ across
the considered elliptical distributions with different kurtosis parameters,
thus, highlighting the sensitivity of the sample correlation estimator
$\hat{\rho}$ to the presence of extreme observations. In contrast,
the sampling densities of $\hat{\gamma}$ and $\hat{\gamma}_{ML}$
are in line with the theoretically predicted results and stable across
all distribution specifications for both considered sample sizes.
Interestingly, the similarity estimators appear to outperform the
sample correlation estimator in terms of efficiency when $T$ is small
and/or the data are more heavy-tailed.

\subsection{\label{subsec:Realized-Similarity-Estimator}Similarity Estimator
under Variance Heterogeneity}

In a more general scenario, where the variables are not necessarily
scale homogeneous, the interpretation of the similarity estimator
is different. Let the scatter matrix of $x=(x_{1},x_{2})^{\prime}$
is such that $\sigma_{1}$ and $\sigma_{2}$ may be non-identical.
Consider quantity $\xi=\frac{2\sigma_{12}}{\sigma_{1}^{2}+\sigma_{2}^{2}}$,
which we refer to as the \textit{coefficient of resemblance}. This
coefficient is closely related to the linear correlation coefficient,
$\rho$, and can be interpreted as an alternative measure of statistical
association between random variables. More particularly, $\rho$ and
$\xi$ are proportionally related,
\begin{equation}
\xi=\frac{2\sigma_{1}\sigma_{2}}{\sigma_{1}^{2}+\sigma_{2}^{2}}\rho=\frac{2c}{1+c^{2}}\rho,\label{eq:cor}
\end{equation}
and the coefficient of proportionality depends only on the relative
scale, $c=\frac{\sigma_{2}}{\sigma_{1}}$. While the signs of $\xi=\xi(c,\rho)$
and $\rho$ are always identical, the magnitude of $\xi$ never exceeds
the magnitude of $\rho$, i.e. $|\xi|\leq|\rho|$, due to the Cauchy-Schwartz
inequality. For a fixed $\rho\neq0$, function $\xi(c,\rho)$ attains
its maximum when $c=1$, i.e. when $\sigma_{1}=\sigma_{2}$. The coefficient
of resemblance $\xi$ can be interpreted as a measure of statistical
similarity between the random variables, where both the correlation
and scales are taken into consideration. Namely, $\xi$ is higher
when the variables tend to exhibit higher correlation, along with
more similar scale parameters (magnitudes).

If $x$ is an elliptical random vector, an important feature of $\xi$
is that, under the Fisher transformation, it becomes the mean of the
transformed resemblance measure, $\phi_{r}=\frac{1}{2}\log\Bigl(\frac{1+r}{1-r}\Bigl)$,
where $r$ introduced in \eqref{eq:res}. This property is reflected
in the following proposition.
\begin{prop}
\label{prop:P2} Assume that $x=(x_{1},x_{2})^{\prime}$ is a bivariate
random vector which follows an elliptical distribution with zero location
and positive-definite scatter matrix $\Sigma$. Denote the measure
of resemblance by $r=\frac{2x_{1}x_{2}}{x_{1}^{2}+x_{2}^{2}}$, and
the corresponding Fisher transformation by $\phi_{r}=\frac{1}{2}\log\Bigl(\frac{1+r}{1-r}\Bigl)$.
Then, $\phi_{r}$ is distributed symmetrically around $\phi_{\xi}$
with $\mathbb{E}(\phi_{r})=\phi_{\xi}=\frac{1}{2}\log\Bigl(\frac{1+\xi}{1-\xi}\Bigl)$,
where $\xi=\frac{2\sigma_{12}}{\sigma_{1}^{2}+\sigma_{2}^{2}}$ is
the coefficient of resemblance, and the variance of $\phi_{r}$ is
given by 
\[
V_{\phi_{r}}=\frac{\pi^{2}}{6}-\sum_{k=1}^{\infty}\frac{\cos2k\vartheta}{k^{2}},
\]
where $\vartheta$ is a function of elements in $\Sigma$ (the exact
expression is provided in the proof).
\end{prop}
It is important to mention that $\phi_{r}$, as well as its distribution
and moments, retain robustness to outliers due to intrinsic insensitivity
of $r$ to the total magnitude of $x$. In contrast to the homoskedastic
case ($\sigma_{1}=\sigma_{2}$), the variance of $\phi_{r}$ does
depend on the underlying scatter matrix of $x$. However, for any
positive definite $\Sigma$, we have that $V_{\phi_{r}}\leq\frac{\pi^{2}}{4}$,
with the equality holds only for $\sigma_{1}=\sigma_{2}$. Therefore,
$V_{\phi_{r}}$ reaches its maximum value under the scale homogeneity,
and has a lower value otherwise.

The results in Proposition \ref{prop:P2} allow to generalize the
asymptotic properties of the similarity estimator $\hat{\gamma}$
for the case of potentially heterogeneous scales. We assume a sample
of random vectors $x_{t}=(x_{1,t},x_{2,t})^{\prime}$, for $t=1,...,T$,
are independent and elliptically distributed with a non-singular scatter
matrix $\Sigma$. Then $\hat{\gamma}$ has the limit distribution
\begin{equation}
\sqrt{T}(\hat{\gamma}-\phi_{\xi})\overset{d}{\rightarrow}\mathcal{N}\Bigl(0,V_{\phi_{r}}\Bigl),\label{eq:gamma_clt_2}
\end{equation}
where $\phi_{\xi}$ is the Fisher transformation of $\xi$ and $V_{\phi_{r}}$
is provided in Proposition \ref{prop:P2}. As a result, $\hat{\gamma}$
is a consistent and robust estimator of the resemblance coefficient
(on the Fisher scale), and its sampling distribution remains stable
across the wide class of elliptical densities for the sample observations.
Similarly, $\hat{\gamma}_{ML}$ estimator is consistent for $\phi_{\xi}$,
however it should be interpreted as a pseudo-ML estimator because
the exact distribution of $\phi_{r}$ depends on the unknown scatter
matrix.

\subsection{Scale-Equalized Estimator of Correlation}

If one wants to estimate the correlation coefficient but the scale
parameters $\sigma_{1}$ and $\sigma_{2}$ are non-identical, the
variables need to be re-scaled such that the new scales become equal.
Denote $\eta_{c}=\log c=\log\frac{\sigma_{2}}{\sigma_{1}}$, then
the corresponding re-scaling is given by $\tilde{x}_{1}(\eta_{c})=e^{\frac{\eta_{c}}{2}}x_{1}$
and $\tilde{x}_{2}(\eta_{c})=e^{-\frac{\eta_{c}}{2}}x_{2}$, and the
similarity estimator constructed on the re-scaled variables becomes
a consistent estimator of $\phi_{\rho}$. The following two-step approach
can be used to re-scale the variables while staying in the robust
register and without affecting the correlation structure. 

At first, construct the variable $w_{t}=\log|x_{2,t}|-\log|x_{1,t}|$,
for $t=1,...,T$, and consider the sample median statistic,
\begin{equation}
\hat{\eta}=\underset{t=1,...,T}{\text{median}}(w_{t})=\underset{\eta\in\mathbb{R}}{\text{argmin}}\sum_{t=1}^{T}|w_{t}-\eta|,\label{eq:median}
\end{equation}
which can also be stated as a solution of a minimization problem.
Anticipating the consistency of $\hat{\eta}$ for $\eta_{c}$, established
in the Proposition \ref{prop:P3} below, we re-scale the sample variables
and construct 

\begin{equation}
\tilde{r}_{t}(\hat{\eta})=\frac{2\tilde{x}_{1,t}(\hat{\eta})\tilde{x}_{2,t}(\hat{\eta})}{\tilde{x}_{1,t}^{2}(\hat{\eta})+\tilde{x}_{2,t}^{2}(\hat{\eta})},\qquad\text{where}\quad\left(\begin{array}{c}
\tilde{x}_{1,t}(\hat{\eta})\\
\tilde{x}_{2,t}(\hat{\eta})
\end{array}\right)=\left(\begin{array}{cc}
e^{\frac{\hat{\eta}}{2}} & 0\\
0 & e^{-\frac{\hat{\eta}}{2}}
\end{array}\right)\left(\begin{array}{c}
x_{1,t}\\
x_{2,t}
\end{array}\right),\quad\text{for}\;t=1,...,T\label{eq:opt_problem}
\end{equation}
Then, the scale-equalized estimator represents the similarity estimator
\eqref{eq:realized_sim}, or its ML variant given in Proposition \ref{prop:P1},
applied to $\tilde{r}_{t}(\hat{\eta})$, for $t=1,...,T$, constructed
with re-scaled observations. The result for the scale-equalized ML
similarity estimator is formalized in the following proposition.
\begin{prop}
\label{prop:P3} For a random sample of bivariate elliptical vectors,
$x_{t}=(x_{1,t},x_{2,t})^{\prime}$, for $t=1,...,T$, with zero location
and positive-definite scatter matrix $\Sigma$, the median estimator
$\hat{\eta}$ given in \eqref{eq:median} is strongly consistent and
asymptotically normal with
\[
\hat{\eta}\xrightarrow{a.s.}\eta_{c}=\log\frac{\sigma_{2}}{\sigma_{1}}\qquad\text{and}\qquad\sqrt{T}(\hat{\eta}-\eta_{c})\overset{d}{\rightarrow}\mathcal{N}\Bigl(0,\frac{\pi^{2}}{4}(1-\rho^{2})\Bigl).
\]
Moreover, the scale-equalized ML similarity estimator, $\hat{\gamma}_{ML*}$,
which solves $\sum_{t=1}^{T}\text{tanh}(\tilde{\phi}_{r,t}(\hat{\eta})-\gamma)=0$,
where $\tilde{\phi}_{r,t}(\hat{\eta})=\frac{1}{2}\log\frac{1+\tilde{r}_{t}(\hat{\eta})}{1-\tilde{r}_{t}(\hat{\eta})}$
with $\tilde{r}_{t}(\hat{\eta})$ constructed on the re-scaled variables
as in \eqref{eq:opt_problem}, has a unique solution almost surely
for every $T\geq3$, is a strongly consistent estimator of $\phi_{\rho}$,
and has the asymptotic distribution
\[
\sqrt{T}(\hat{\gamma}_{ML*}-\phi_{\rho})\overset{d}{\rightarrow}\mathcal{N}\Bigl(0,2\Bigl),
\]
where $\phi_{\rho}=\frac{1}{2}\log\Bigl(\frac{1+\rho}{1-\rho}\Bigl)$
is the Fisher transformation of $\rho$.
\end{prop}
The finite-sample properties of the scale-equalized similarity estimator
$\hat{\gamma}_{ML*}$ are illustrated in Figure \ref{fig:finite_sample_densities_mod}.
We consider the same underlying distributions and sample sizes as
in \ref{fig:finite_sample_densities}, but now with heterogeneous
scale parameters, $\sigma_{2}^{2}=4\sigma_{1}^{2}$. The results confirm
that the finite-sample distribution of $\hat{\gamma}_{ML*}$ quickly
converges to the asymptotic limit being sufficiently close to it already
for $T=8$ and almost identical to it for $T=40$. As it is expected
from the given robust estimation design, the sampling distribution
of $\hat{\gamma}_{ML*}$ is remarkably stable across different underlying
data distributions even for very small $T$. 

\begin{figure}[h]
\begin{centering}
\includegraphics[scale=0.8]{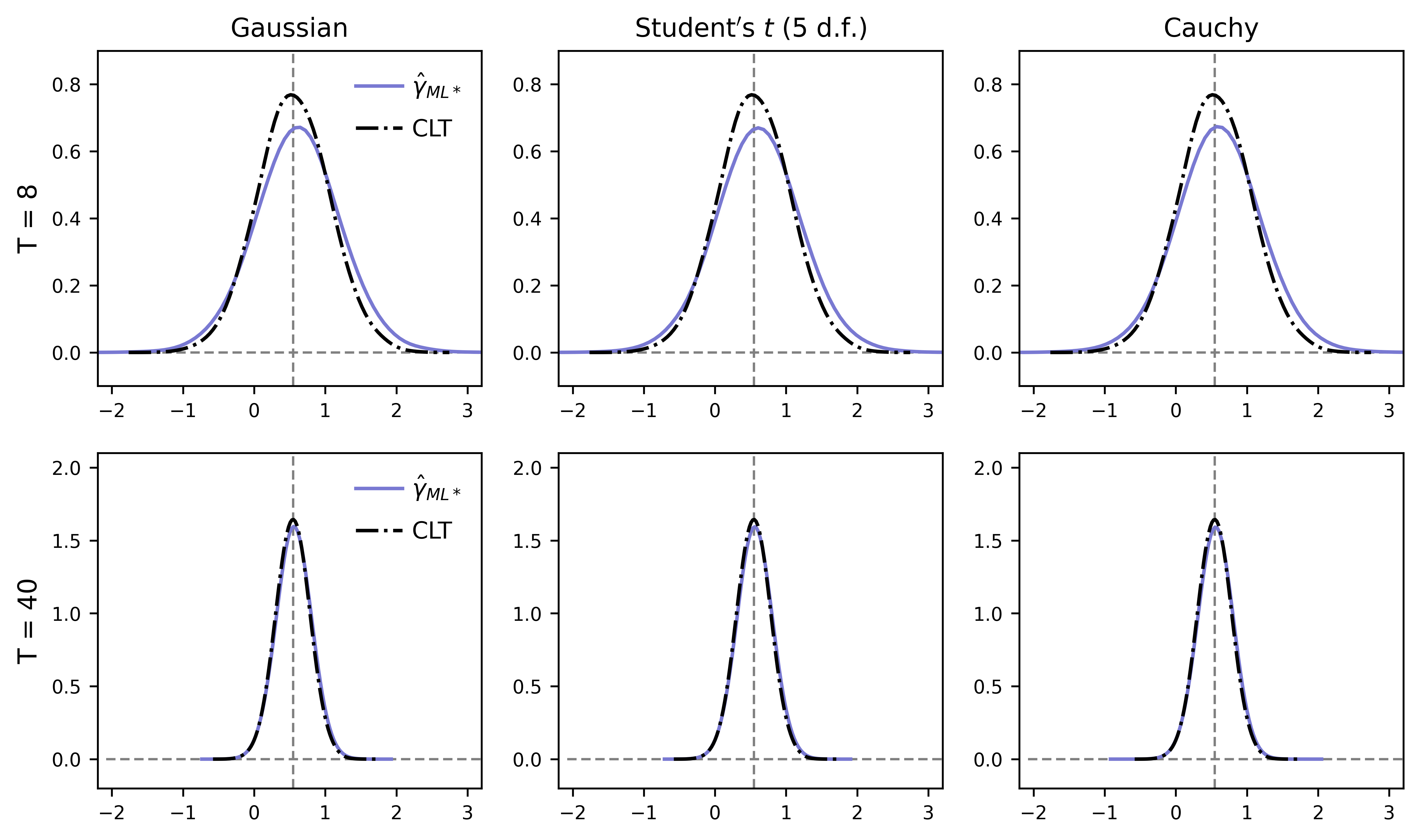}
\par\end{centering}
\caption{\footnotesize\label{fig:finite_sample_densities_mod} Finite-sample
distributions of the scale-equalized ML similarity estimator $\hat{\gamma}_{ML*}$
(blue solid lines) obtained on $100,000$ simulated samples of sizes
$T=8$ (top plots) and $T=40$ (bottom plots). Data vectors $x_{t}$
were simulated out of the three selected bivariate distributions --
normal, $t$-distribution with $5$ degrees of freedom, and Cauchy
-- with the true correlation parameter $\rho=0.5$ and with heterogeneous
variances, $\sigma_{2}^{2}=4\sigma_{1}^{2}$. Black dash-dotted lines
show the distribution of $\hat{\gamma}_{ML*}$ predicted by the asymptotic
result, $\sqrt{T}(\hat{\gamma}_{ML*}-\phi_{\rho})\xrightarrow{d}\mathcal{N}\bigl(0,2\bigl)$.}
\end{figure}

Obviously, the suggested scheme for re-scaling is not a unique possible
method. Alternatively, the variables can be re-scaled by using any
robust scale or variance estimators. Once $\hat{\sigma}_{1}$ and
$\hat{\sigma}_{2}$ are consistent (and robust) estimates of the corresponding
scale parameters, the standardized observations can be constructed,
$z_{t}=\Bigl(\frac{x_{1,t}}{\hat{\sigma}_{1}},\frac{x_{2,t}}{\hat{\sigma}_{2}}\Bigl)^{\prime}$.
In the second step, the similarity estimator, $\hat{\gamma}$ or $\hat{\gamma}_{ML}$,
is applied to the standardized observations $z_{t}$.

\section{\label{sec:Measuring-Similarity-for}Measuring Similarity for Multiple
Variables}

The quantity $\phi_{r,t}$ formulated in \eqref{eq:phi_r} can be
represented as a log-ratio of the \textit{local} dispersion measures
of vector $x_{t}$ -- variation of the sum and variation of the difference
of the vector components, where $x_{t}=(x_{1,t},x_{2,t})^{\prime}$
is a zero location vector. For fixed $\sigma_{1}^{2}$ and $\sigma_{2}^{2}$,
an increase in $\rho$ makes the first variation higher and the second
variation lower, and conversely. Naturally, for $\rho\approx\pm1$
the contrast between the variation measures is the highest, while
for $\rho=0$ they are identical. Therefore, $\phi_{r,t}$ can be
viewed as a local measure of the divergence between quadratic forms
$(x_{1,t}+x_{2,t})^{2}$ and $(x_{1,t}-x_{2,t})^{2}$ which nests
information about the correlation coefficient $\rho$, and this information
can be explicitly recovered once $\sigma_{1}^{2}=\sigma_{2}^{2}$,
as shown in Section \ref{subsec:Variance-Homogeneity:-Robust}. 

We note that $(x_{1,t}+x_{2,t})^{2}$ can alternatively be represented
as the squared projection of $x_{t}$ onto the direction of the vector
of ones, i.e. the direction of perfect\textit{ similarity}. Analogously,
$(x_{1,t}-x_{2,t})^{2}$ is the squared projection of $x_{t}$ onto
the orthogonal direction, i.e. the direction of \textit{dissimilarity}.
Therefore, $\phi_{r,t}$, and consequently $\hat{\gamma}$, indicate
the degree of similarity in the joint variation of $x_{1,t}$ and
$x_{2,t}$ measured by the relative magnitude of variation along the
vector of ones. This idea can be directly generalized to random vectors
of arbitrarily large dimension.

Assume that $x_{t}=(x_{1,t},x_{2,t},...,x_{n,t})^{\prime}$ is a zero
location $n$-dimensional elliptical vector with positive definite
scatter matrix $\Sigma$. We denote the $n$-dimensional vector of
ones by $\iota_{n}=(1,1,...,1)^{\prime}$. The matrix $P_{n}=\frac{1}{n}\iota_{n}\iota_{n}^{\prime}$
is the projection matrix such that the projection of $x_{t}$ onto
the direction defined by the vector $\iota_{n}$ is given by $P_{n}x_{t}$.
Then, $P_{n}^{\perp}=I_{n}-P_{n}$, where $I_{n}$ is the identity
matrix of dimension $n$, is also the projection matrix such that
$P_{n}^{\perp}x_{t}$ represents the projection of $x_{t}$ onto the
subspace orthogonal to $\iota_{n}$. Then, the local similarity measure
$\phi_{r,t}$ can be naturally extended to the general multivariate
case via 
\begin{equation}
\phi_{r,t}=\frac{1}{n}\log\frac{x_{t}^{\prime}P_{n}x_{t}}{x_{t}^{\prime}P_{n}^{\perp}x_{t}},\label{eq:local_sim_mult}
\end{equation}
which captures the variation of $x_{t}$ along the vector of \textit{perfect
similarity}, $\iota_{n}$, relative to the variation of $x_{t}$ along
its orthogonal complement. Note that, for the special case $n=2$,
equation \eqref{eq:local_sim_mult} is reduced to the same expression
for $\phi_{r,t}$ as in \eqref{eq:phi_r}. 

As in the bivariate case, the similarity estimator $\hat{\gamma}$
is constructed as $\hat{\gamma}=\frac{1}{T}\sum_{t=1}^{T}\phi_{r,t}$.
Intuitively, $\hat{\gamma}$ represents an aggregate measure of how
well $n$ considered variables agree simultaneously in magnitude and
direction. 

\subsection{\label{subsec:Equicorrelation-Scenario} Equicorrelation Scenario}

In the multivariate scenario with $n>2$, the similarity estimator
$\hat{\gamma}$ has, in general, less tractable interpretation as
compared to the bivariate setting. In some special cases, however,
$\hat{\gamma}$ retains explicit finite-sample and asymptotic limits,
and is immediately related to the linear correlation coefficient. 

Let the entries of the scatter matrix $\Sigma$ are denoted by 
\[
\Sigma=\left(\begin{array}{cccc}
\sigma_{1}^{2} & \cdot & \cdot & \cdot\\
\sigma_{12} & \sigma_{2}^{2} & \cdot & \cdot\\
\cdot & \cdot & \ddots & \cdot\\
\sigma_{1n} & \sigma_{2n} & \cdot & \sigma_{n}^{2}
\end{array}\right),
\]
and here we assume that all scale parameters are equal, i.e. $\sigma_{1}^{2}=\sigma_{2}^{2}=...=\sigma_{n}^{2}=\sigma^{2}$,
and all off-diagonal scatter elements are identical too. This implies
that all pairwise correlations are also identical, and, so, we can
write $\sigma_{ij}=\sigma^{2}\rho$ for all $i$ and $j$, such that
$i\neq j$, where $\rho$ is the common correlation parameter. Note
that, once $\Sigma$ is positive definite by assumption, it implies
that $\rho\in\bigl(-\frac{1}{n-1},1\bigl)$. 

The matrix logarithm transformation suggests a convenient method for
reparametrization of correlation matrices in such way that the transformed
correlation elements can be represented as an unconstrained real vector
which can always be mapped to a unique positive definite correlation
matrix. In this sense, the parametrization based on the matrix logarithm
can be viewed as a multi-dimensional generalization of the Fisher
transformation, see \citet{Archakov_Hansen_2021} for more details.
When the matrix logarithm transformation is applied to $\Sigma$ with
homogeneous variances and equal correlations, the general structure
of the matrix is preserved after the transformation, i.e. all diagonal
and all off-diagonal elements remain identical, 

\setlength{\arraycolsep}{5pt}
\[
\Sigma=\sigma^{2}\left(\begin{array}{ccccc}
1 & \cdotp & \cdotp & \cdot & \cdotp\\
{\color{purple}{\footnotesize \rho}} & 1 & \cdotp & \cdot & \cdotp\\
{\color{purple}{\footnotesize \rho}} & {\color{purple}{\footnotesize \rho}} & 1 & \cdot & \cdotp\\
\vdots & \cdot & \vdots & \ddots & \cdot\\
{\color{purple}{\footnotesize \rho}} & {\color{purple}{\footnotesize \rho}} & {\color{purple}{\footnotesize \rho}} & \cdots & 1
\end{array}\right),\qquad\log\Sigma=\left(\begin{array}{ccccc}
\delta & \cdotp & \cdotp & \cdot & \cdotp\\
{\color{blue}\phi_{\rho}} & \delta & \cdotp & \cdot & \cdotp\\
{\color{blue}\phi_{\rho}} & {\color{blue}\phi_{\rho}} & \delta & \cdot & \cdotp\\
\vdots & \cdot & \vdots & \ddots & \cdot\\
{\color{blue}\phi_{\rho}} & {\color{blue}\phi_{\rho}} & {\color{blue}\phi_{\rho}} & \cdots & \delta
\end{array}\right),
\]
as it is shown in \citet{Archakov_Hansen_2024}. The eigenvalues of
matrix $\Sigma$ are $\lambda_{+}=\sigma^{2}\bigl(1+(n-1)\rho\bigl)$
and $\lambda_{-}=\sigma^{2}(1-\rho)$ with multiplicities $1$ and
$n-1$, respectively, see \citet{Olkin_Pratt_1958}. The entries of
$\log\Sigma$ are analytically available and given by $\delta=\frac{1}{n}\log\lambda_{+}+\frac{n-1}{n}\log\lambda_{-}$
and 
\begin{equation}
\phi_{\rho}=\frac{1}{n}\log\frac{\lambda_{+}}{\lambda_{-}}=\frac{1}{n}\log\Bigl(\frac{1+(n-1)\rho}{1-\rho}\Bigl).\label{eq:gamma_equi}
\end{equation}
We note that the off-diagonal element $\phi_{\rho}$ depends only
on $\rho$ and $n$, and does not depend on $\sigma^{2}$. Naturally,
this result is a generalization of the matrix logarithm result for
$n=2$ in Section \ref{subsec:Variance-Homogeneity:-Robust}, so we
preserve notation $\phi_{\rho}$ introduced earlier. Based on this
identity, the distribution of $\phi_{r,t}$ in \eqref{eq:local_sim_mult}
can be derived explicitly which is reflected in the following result.
\begin{prop}
\label{prop:P4} Assume that $x=(x_{1},x_{2},...,x_{n})^{\prime}$
is a $n$-variate elliptical vector with zero location and positive-definite
scatter matrix $\Sigma$ with homogeneous scales and identical correlations
(equi-correlation structure), and let denote the common correlation
coefficient by $\rho$. Then, $\phi_{r}=\frac{1}{n}\log\frac{x^{\prime}P_{n}x}{x^{\prime}P_{n}^{\perp}x}$
is a random variable from the Logistic-Beta family with the probability
density function
\[
f(\phi_{r})=\frac{1}{\mathcal{B}\Bigl(\frac{1}{2},\frac{n-1}{2}\Bigl)}\cdot\frac{ne^{\frac{1}{2}n(\phi_{r}-\phi_{\rho})}}{\Bigl(1+e^{n(\phi_{r}-\phi_{\rho})}\Bigl)^{\frac{n}{2}}},
\]
where $\phi_{\rho}=\frac{1}{n}\log\Bigl(\frac{1+(n-1)\rho}{1-\rho}\Bigl)$
is the off-diagonal element of $\log\Sigma$. 
\end{prop}
Note that $\phi_{r}$ is not an unbiased measure of $\phi_{\rho}$
because $\mathbb{E}(\phi_{r})=\phi_{\rho}-\omega_{n}$, where $\omega_{n}\ge0$
is a non-monotone function of $n$ such that $\omega_{2}=0$ and $\omega_{n}\rightarrow0$
as $n\rightarrow\infty$. The exact expression for $\omega_{n}$ is
provided in Appendix. Variance of $\phi_{r}$ is inversely proportional
to $n^{2}$ and is given by 
\[
\mathbb{V}(\phi_{r})=\frac{1}{n^{2}}\biggl(\psi^{\prime}\Bigl(\frac{n-1}{2}\Bigl)+\frac{\pi^{2}}{2}\biggl),
\]
where $\psi^{\prime}(t)$ is the trigamma function. Naturally, when
vector dimension $n$ increases, $\phi_{r}$ absorbs more information
about the common correlation coefficient from the cross-sectional
dimension and, so, provides an increasingly more efficient signal
about the latent correlation level. In Figure \ref{fig:pdf2}, we
illustrate the Logistic-Beta probability density function of $\phi_{r}-\phi_{\rho}$
for several selected values of $n$. 

In the special case of homogeneous scales and correlations, Proposition
\ref{prop:P4} helps to characterize the asymptotic properties of
the similarity estimator $\hat{\gamma}$. With a bias correction,
$\hat{\gamma}+\omega_{n}$ becomes a consistent and robust estimator
for the off-diagonal element $\phi_{\rho}$ of the log-transformed
scatter matrix $\Sigma$, where $\phi_{\rho}$ represents a monotone
transformation of the equicorrelation coefficient $\rho$, similarly
to the Fisher transformation function in the bivariate case. The asymptotic
variance of $\hat{\gamma}$ is given by $\mathbb{V}(\phi_{r})$, and
the estimator becomes more efficient with dimension $n$ due to the
growth of relevant cross-sectional information. 

\begin{figure}[h]
\begin{centering}
\includegraphics[scale=0.8]{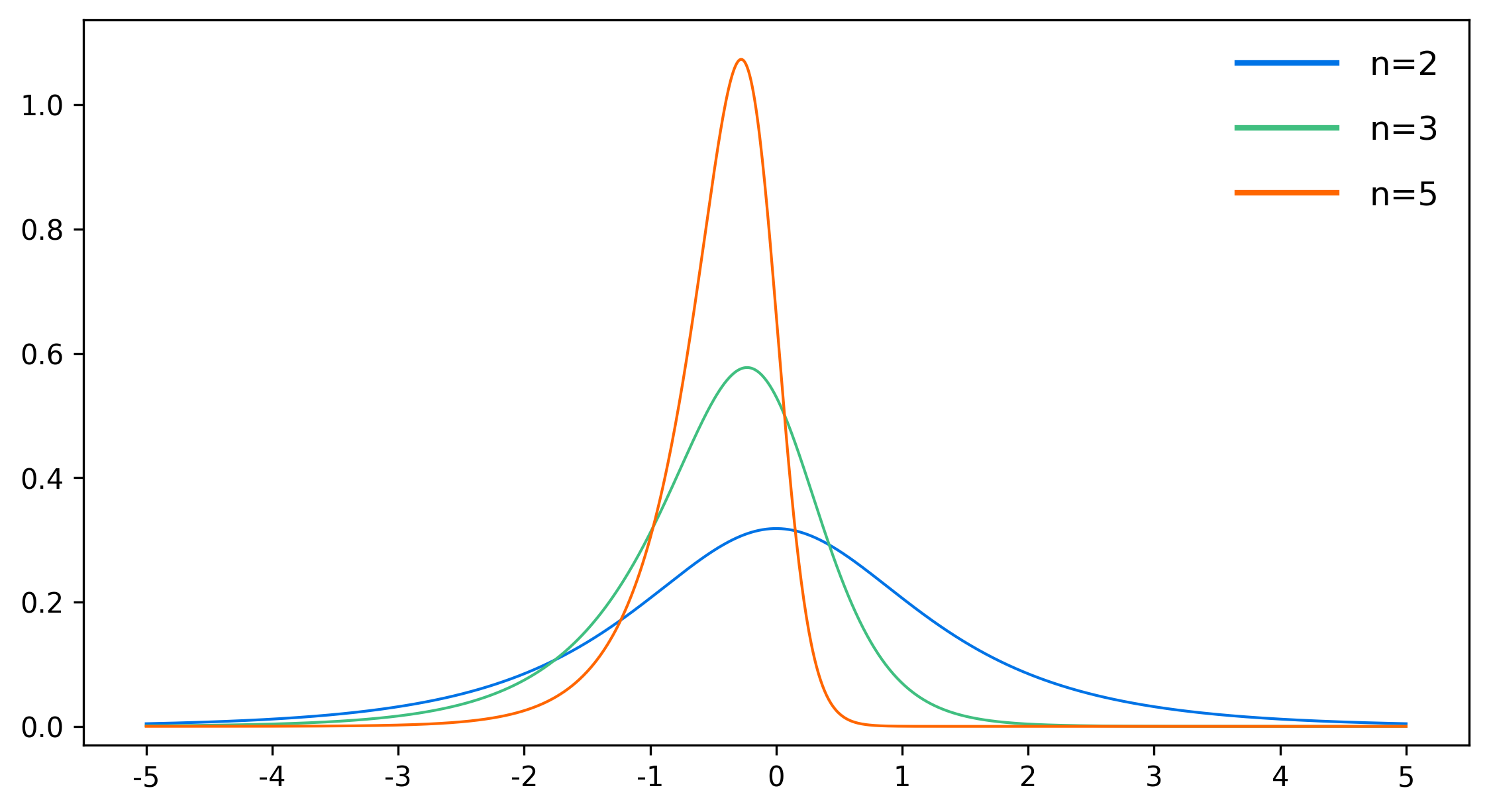}
\par\end{centering}
\caption{\footnotesize\label{fig:pdf2} Probability density functions of $\phi_{r}-\phi_{\rho}$
for different dimensions $n$ of vector $x$.}
\end{figure}

Despite the structure assumed for $\Sigma$ is restrictive and does
not occur in practice frequently, the results reveal several qualitative
aspects which are characteristic of the similarity estimator, $\hat{\gamma}$,
in more general and realistic scenarios. At first, $\hat{\gamma}$
is a robust estimator of an aggregate measure of association between
a set of multiple variables. Although this does not directly correspond
to the average correlation level for an arbitrary covariance or scatter
matrix $\Sigma$, it still can be interpreted as an indicator of \textit{joint}
statistical similarity between variables, i.e. how well they agree
in both magnitude and direction. At second, if elements in $\Sigma$
are sufficiently homogeneous, $\hat{\gamma}$ is expected to be more
efficient once $n$ gets larger. This is because each additional variable
provides extra information about the average level of similarity. 

\section{Empirical Applications}

Although the similarity estimator can be used as a standalone measure
of statistical association, in this section, we consider empirical
applications where the estimator is used for estimation and modeling
of correlations. The two considered applications focus on measuring
correlations between financial stock returns. 

\subsection{Robust Interval Estimation of Stock Return Correlations }

We apply the similarity estimator for robust interval estimation of
financial correlations. For this analysis, we use intraday transaction
data from the TAQ database cleaned according to the recommendations
provided in \citet{Barndorff-Nielsen2009}. For exposition purposes,
we begin with considering daily correlations between Apple (AAPL)
and Exxon Mobil (XOM) returns during the period between February and
April 2020 (62 trading days) which corresponds to the COVID-19 outbreak. 

We note that AAPL and XOM belong to different sectors of the economy
(Tech vs Energy) implying the correlation is not supposed to be particularly
strong in ordinary periods. During the COVID-19 crisis, however, the
overall correlation level has significantly elevated across the entire
market due to i) an initial decline of the market that affected almost
all sectors and industries (largely began on February 19), and ii)
the subsequent common recovery (began after March 23). Therefore,
we may expect some prominent dynamics of the latent correlation trajectory
during the analyzed period.

For each considered trading day, we construct samples of intraday
(logarithmic) returns, $r_{t}=(r_{1,t},r_{2,t})^{\prime}$, $t=1,...,T$,
for a range of selected frequencies by using the previous tick interpolation
scheme which was introduced in \citet{Wasserfallen_Zimmermann_1985}
and is used routinely in high frequency econometrics (see, for example,
\citet{Hansen_Lunde_2006}). Since the duration of a typical trading
day is $6.5$ hours, we obtain a sample with $T=39$ observations
for the frequency $\Delta=10$ min, while we have $T=780$ observations
for $\Delta=30$ sec. 

\begin{figure}[!ph]
\begin{centering}
\subfloat[Intraday returns are observed at $\Delta=10$ min frequency ($T=39$
observations).]{\centering{}\includegraphics[scale=0.52]{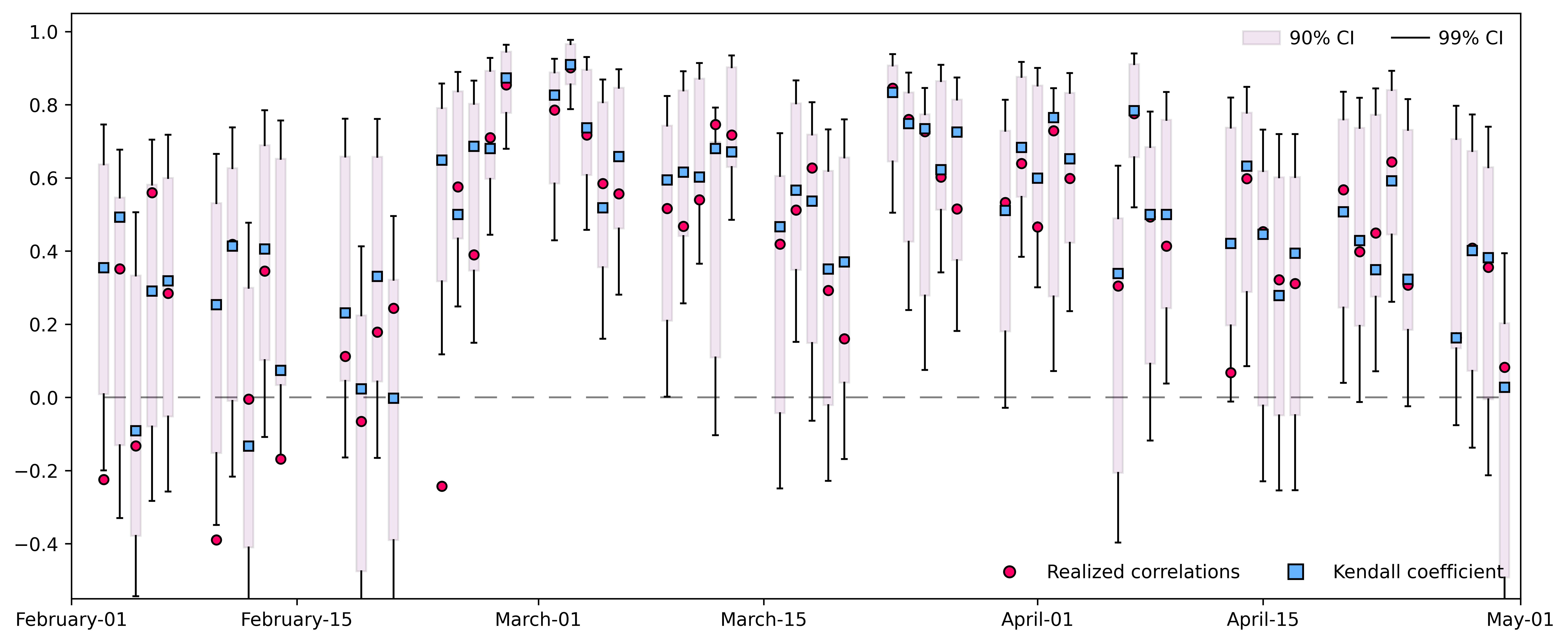}}
\par\end{centering}
\begin{centering}
\subfloat[Intraday returns are observed at $\Delta=2$ min frequency ($T=185$
observations).]{\centering{}\includegraphics[scale=0.52]{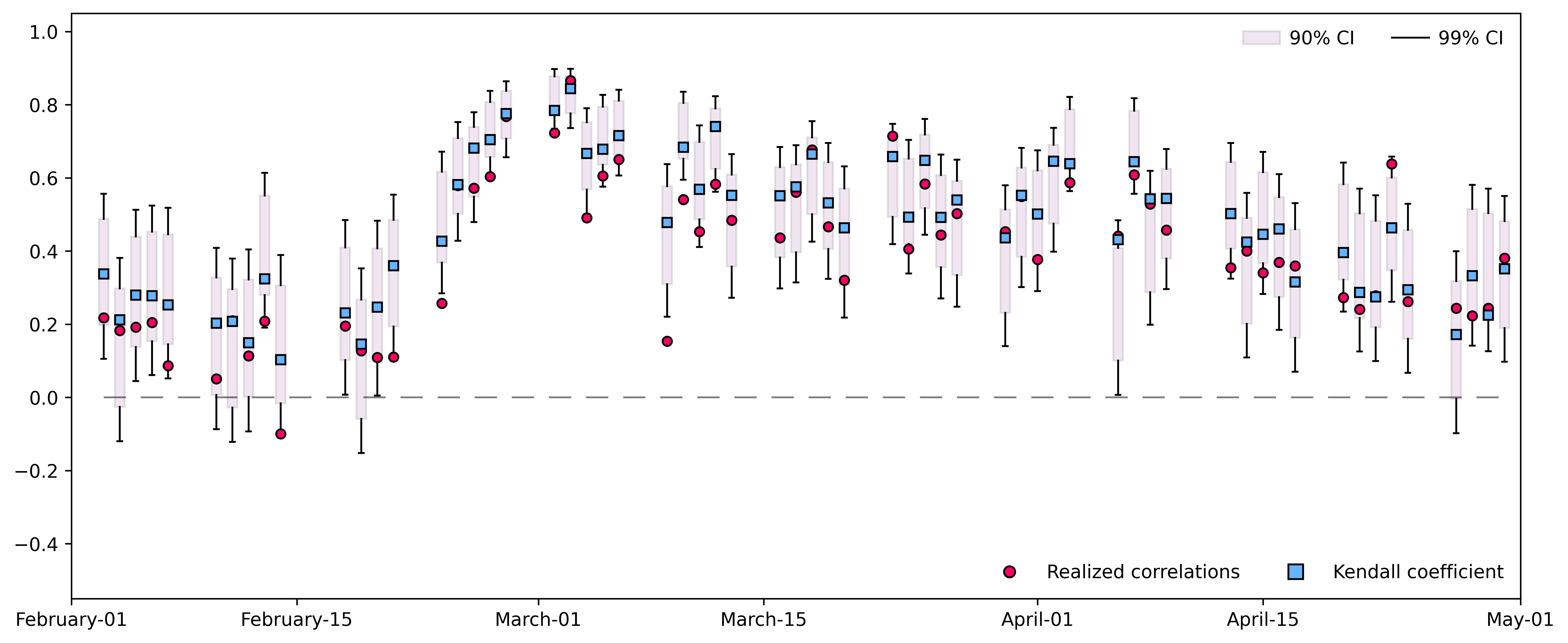}}
\par\end{centering}
\centering{}\subfloat[Intraday returns are observed at $\Delta=30$ sec frequency ($T=780$
observations).]{\centering{}\includegraphics[scale=0.52]{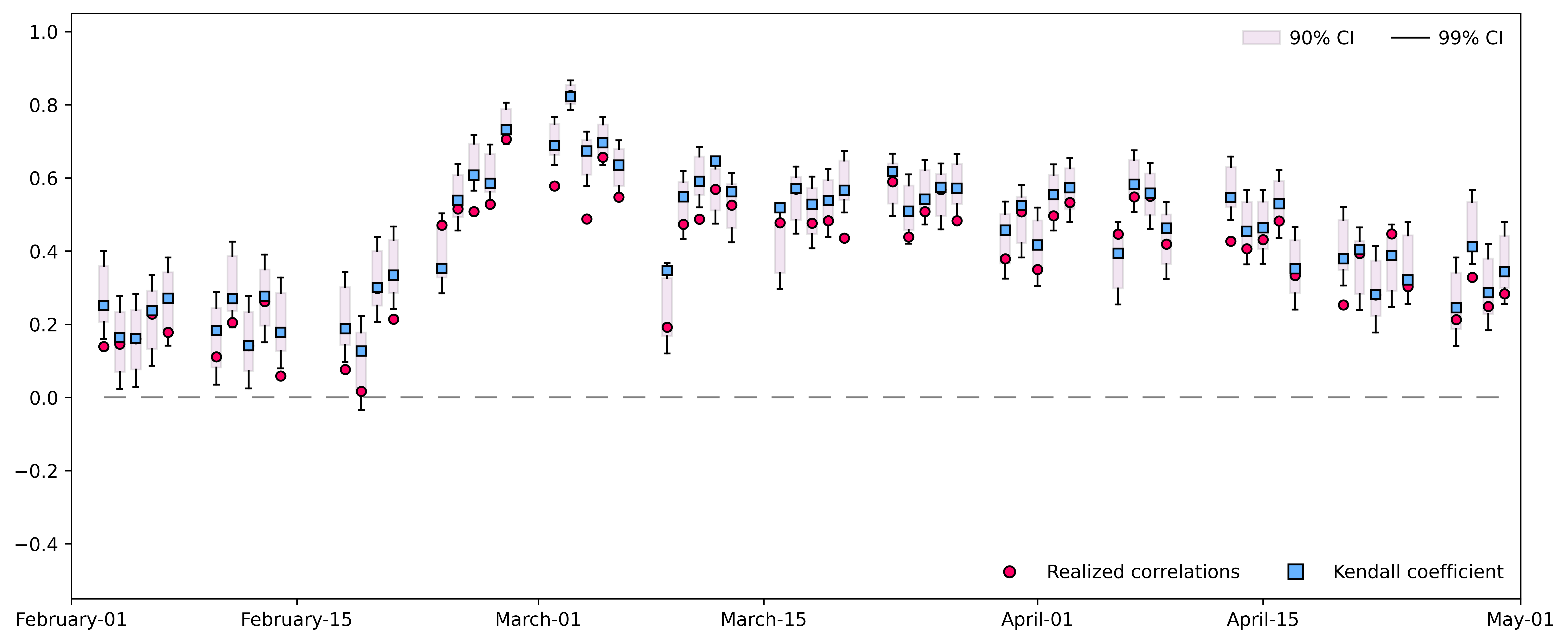}}\caption{\footnotesize\label{fig:spring_2020} Confidence intervals for daily
correlations between Apple (AAPL) and Exxon Mobil (XOM) estimated
using intraday returns. The analyzed period is between February and
April 2020 (62 trading days). The robust intervals are obtained for
coverage probabilities 90\% (boxes) and 99\% (whiskers). Red circles
correspond to daily Realized (sample) Correlations, while blue squares
correspond to daily correlations estimated by the Kendall tau coefficient. }
\end{figure}

To construct an interval with a specified coverage probability, we
estimate the correlation coefficient (on the Fisher scale) using the
scale-equalized ML similarity estimator $\hat{\gamma}_{ML*}$. Next,
we obtain the required interval around $\hat{\gamma}_{ML*}$ using
quantiles of the asymptotic distribution given in Proposition \ref{prop:P3},
and map the interval from the Fisher scale to the standard correlation
scale. The resulting intervals for daily correlations are obtained
for 90\% and 99\% coverage probabilities for the three selected intraday
frequencies ($\Delta=10$ min, $2$ min, and $30$ sec) and are shown
in Figure \ref{fig:spring_2020}. The figure also provides point correlation
estimates calculated with two benchmark estimators -- the sample
(realized) correlation estimator and the Kendall rank coefficient
(see Section \ref{sec:Realized-Similarity-Estimator} for the corresponding
expressions). The results in Figure \ref{fig:spring_2020} reveal
several interesting aspects.

When constructed with returns sampled at $\Delta=10$ min frequency,
the intervals are wide and thus not very informative about the underlying
correlation level. For the majority of trading days, the intervals
do not even allow to reject the null hypothesis of zero correlation.
Such conservative interval widths reflect the trade-off required to
achieve robustness and invariance of the constructed intervals under
arbitrary fat-tailed data. At higher frequencies, where more observations
are available, intervals naturally shrink, and the evolution of the
correlation level over time becomes more apparent. Thus, when constructed
at $\Delta=30$ sec frequency, the intervals are sufficiently narrow
allowing to clearly observe a sharp surge of the correlation in the
second half of February and gradual non-monotone subsequent decay.
Interestingly, the estimated intervals appear to cluster within weeks,
which may indicate highly persistent correlation dynamics.

The robust intervals based on the similarity estimator show good agreement
with the benchmark correlation estimators for almost all trading days
at moderate and low sampling frequencies, where the intervals are
sufficiently wide. For a number of trading days, however, we observe
that the sample correlations fall below the constructed intervals
as well as below the estimates obtained with the Kendall estimator.
This can be explained by the common presence of price jumps in the
market data inducing outliers in high-frequency returns. While the
realized correlation estimator exhibits a downward bias in such cases,
the Kendall and similarity estimators remain robust.

\begin{figure}[ph]
\begin{centering}
\includegraphics[scale=0.75]{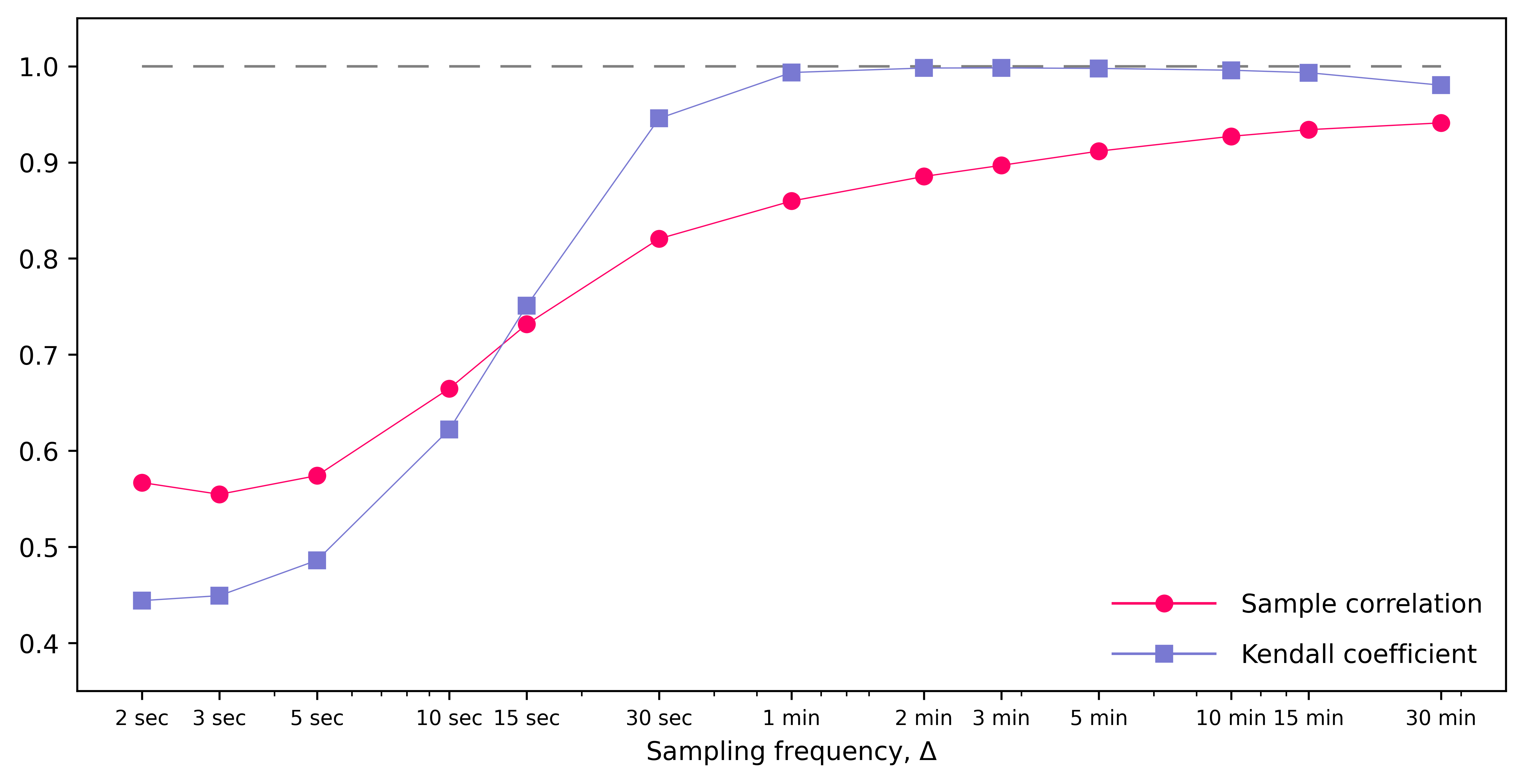}
\par\end{centering}
\caption{\footnotesize\label{fig:estimates_match} Proportions of daily realized
correlations (red circle marks) and Kendall correlation coefficients
(blue square marks) which lie inside the robust 95\% intervals constructed
using the robust similarity estimator. The results are based on $845,670$
estimated daily pairwise correlations. Both the correlation estimates
and the robust intervals are constructed with intra-daily asset returns
at $13$ different sampling frequencies.}
\end{figure}

We also implement a more systematic analysis in order to inspect how
well the proposed robust intervals agree with the standard estimators
when applied to financial correlations. For this we use intraday returns
for $21$ liquid assets from different sectors (Energy, Industrials,
Materials, Consumer Staples, Finance, Information Technology, and
Communication Services) for the period between January 2005 and December
2020 comprising 4,027 trading days. In total, our analysis is based
on $845,670$ pairwise daily correlations, and for each we consider
13 different sampling frequencies ranging from $\Delta$=2 sec ($T$=11,700
observations) to $\Delta$=30 min ($T$=13 observations). For each
frequency $\Delta$, we calculate the proportion of realized correlations
and Kendall correlations which fall into 95\% robust intervals estimated
with $\hat{\gamma}_{ML*}$.

The results are presented in Figure \ref{fig:estimates_match}. Thus,
for moderate sampling frequencies ($\Delta\geq$30 sec), the correlations
obtained with the robust Kendall estimator show extremely good agreement
with the robust intervals, and the proportion of matches approaches
$1$. In contrast, for the same range of $\Delta$, the corresponding
proportion for the realized correlation estimator is sufficiently
lower (up to $20$ percentage points below the corresponding results
of the Kendall estimator) and tends to decrease steadily as the frequency
increases. This pattern is intuitive due to intrinsically non-robust
design of the sample correlation estimator. A poor agreement exhibited
by both estimators at ultra high frequencies ($\Delta\leq$15 sec)
is expected. The market microstructural effects, such as price discreteness
(due to fixed increment size) or price staleness (due to insufficient
liquidity), become especially prominent relative to the genuine price
signals (see \citet{Hansen_Lunde_2006}, \citet{bandi2017excess},
among others), thus compromising the elliptical assumption for the
return distribution. As a result, all considered estimators may converge
to different quantities. Given that the estimated intervals get more
narrow at higher frequencies because of high $T$, it becomes increasingly
more likely that the point correlation estimates occur outside these
intervals. Based on this evidence, we cautiously conclude the robust
similarity estimator works reliably with intraday returns of sufficiently
liquid stocks at moderately high frequencies ($\Delta\geq$30 sec),
while may be unstable at ultra high frequencies.

We emphasize that the presented analysis serves rather for experimental
and illustration purposes. The assumption of independent and identically
distributed observations can hardly be justified due to stochastic
volatility, market microstructural effects, and other stylized artifacts
which are typically attributed to high frequency returns. Moreover,
we may expect that the correlation level change over a trading day,
therefore the estimates should rather be interpreted as measures of
some average daily correlation. A more careful adaptation of the similarity
estimator for high frequency financial data is a challenging and interesting
question which is a subject of ongoing work.

\subsection{A Robust Multivariate GARCH Model}

We develop a new specification of the multivariate GARCH model that
naturally accommodates the empirical similarity measure introduced
in Section \ref{sec:Measuring-Similarity-for}. Let $R_{t}=(R_{1,t},R_{2,t},...,R_{n,t})^{\prime}$
be an $n$-dimensional vector of asset returns, for $n\geq2$, observed
at discrete time moments $t=1,...,T$, and let $\{\mathcal{F}_{t}\}$
denote the natural filtration for $R_{t}$. We assume that, conditionally
on $\mathcal{F}_{t-1}$, the return vector follows an elliptical distribution
with a constant location vector $\mu=(\mu_{1},\dots,\mu_{n})'$ and
a positive definite conditional \textit{scatter} matrix $S_{t}$,
which is the key object of interest. Following the logic of the Dynamic
Conditional Correlation (DCC) approach of \citet{Engle_2002}, we
decompose $S_{t}$ as
\[
S_{t}=\Lambda_{t}^{1/2}C_{t}\Lambda_{t}^{1/2},
\]
where $\Lambda_{t}=\text{diag}(h_{1,t},h_{2,t},...,h_{n,t})^{\prime}$
collects the conditional (squared) scale parameters of the individual
returns, $\ensuremath{h_{i,t}=[S_{t}]_{ii}}$, for $i=1,...,n$, and
$C_{t}$ is the positive definite matrix of conditional scatter-correlation
parameters implied by $S_{t}$. Whenever the conditional distribution
possesses finite second moments, the conditional covariance matrix
exists and is proportional to the scatter,
\[
\big(R_{t}\big|\mathcal{F}_{t-1}\big)=m_{g}S_{t},
\]
where the constant $m_{g}>0$ is determined by the elliptical generator
($m_{g}=1$ in the Gaussian case, and $m_{g}=\nu/(\nu-2)$ for the
Student's $t$ with $\nu>2$ degrees of freedom). In that case, $C_{t}$
coincides with the conditional correlation matrix $\text{corr}(R_{t}|\mathcal{F}_{t-1})$,
and $h_{i,t}$ is proportional to the conditional variance of $R_{i,t}$.
For heavier-tailed generators without second moments, such as the
symmetric multivariate stable, the scatter matrix $S_{t}$ and the
matrix $C_{t}$ remain well defined, and $C_{t}$ retains its interpretation
as the matrix of correlation parameters of the elliptical family.
This structure allows us to split the modeling of $S_{t}$ into the
separate modeling of the conditional scales and of the conditional
correlations.

We assume that all conditional locations are constant, which is a
standard assumption in the GARCH literature, and formulate the return
equations as
\begin{equation}
R_{i,t}=\mu_{i}+h_{i,t}^{1/2}z_{i,t},\qquad i=1,...,n,\quad\quad t=1,...,T,\label{eq:ret_eq}
\end{equation}
where $z_{t}=(z_{1,t},z_{2,t},...,z_{n,t})^{\prime}=\Lambda_{t}^{-\frac{1}{2}}(R_{t}-\mu)$
denotes the vector of standardized returns. Conditionally on $\mathcal{F}_{t-1}$,
$z_{t}$ is elliptical with zero location and scatter matrix $C_{t}$;
in particular, the components of $z_{t}$ have \textit{unit conditional
scales}, so the conditional scatter of $z_{t}$ is homogeneous by
construction (and each $z_{i,t}$ has conditional variance $m_{g}$,
common across $i$, whenever second moments exist). Therefore, $z_{t}$
carries the information about the conditional correlation structure
of the raw returns and, under the equicorrelation assumption imposed
below, satisfies exactly the assumptions of Proposition \ref{prop:P4},
conditionally on $\mathcal{F}_{t-1}$, with no moment requirements
on the generator. The conditional scale parameters $h_{i,t}$, for
$i=1,\dots,n$, can be modeled with any appropriate univariate GARCH-type
filter. In our empirical application, we use the EGARCH(1,1) model
by \citet{Nelson_1991}. Since daily returns are typically fat-tailed,
the conditional elliptical family may be chosen to allow for excess
kurtosis, such as the Student's $t$ or the symmetric stable distribution.\footnote{The elliptical assumption leaves the radial component unrestricted;
a parametric choice (e.g., the degrees of freedom of a Student's $t$)
is needed only if full-likelihood estimation of the univariate filters
is desired.}

In this paper, we restrict our analysis to an equicorrelation assumption
for $C_{t}$, such that all conditional correlations are identical
and parametrized by a single dynamic coefficient, $\rho_{t}$. This
structure implies a very parsimonious model specification which is
readily scaled for arbitrary large dimension $n$ and allows to elegantly
employ results in Section \ref{subsec:Equicorrelation-Scenario}.
The equicorrelation assumption for $C_{t}$ was originally introduced
in \citet{Engle_Kelly_2012} (known as the DECO model), and was extended
to accommodate realized measures of variances and correlations in
\citet{Archakov_Hansen_Lunde_2025}. 

The methodological novelty is related to the way how the dynamics
for conditional correlation matrix $C_{t}$ is modeled. Following
the approach in \citet{Archakov_Hansen_Lunde_2025}, we set up the
dynamics for the off-diagonal elements of the log-transformed conditional
correlation matrix, $\log C_{t}$. The novel idea is to use the similarity
measure $\phi_{r,t}$, characterized in Section \ref{sec:Measuring-Similarity-for},
as a natural signal about the current level of $\log C_{t}$. Since,
conditionally on $\mathcal{F}_{t-1}$, $z_{t}$ is elliptical with
homogeneous unit scales and equicorrelation matrix $C_{t}$, Proposition
\eqref{prop:P4} applies period by period, and $\phi_{r,t}$ constructed
with $z_{t}$ represents a robust empirical measure of the transformed
correlation coefficient.

For an equicorrelation matrix $C_{t}$, the off-diagonal entries of
$\log C_{t}$ are all identical and equal to $\phi_{\rho,t}$, and
the analytical relationship between $\phi_{\rho,t}$ and $\rho_{t}$
is provided in \eqref{eq:gamma_equi}. We formulate dynamics for $\phi_{\rho,t}$
as follows, 
\begin{equation}
\phi_{\rho,t}=\alpha+\beta\cdot\phi_{\rho,t-1}+(\kappa+\varsigma\cdot1_{\{\iota_{n}^{\prime}z_{t-1}<0\}})\cdot\underbrace{\frac{1}{n}\log\frac{z_{t-1}^{\prime}P_{n}z_{t-1}}{z_{t-1}^{\prime}P_{n}^{\perp}z_{t-1}}}_{=\phi_{r,t-1}},\label{eq:deco_garch}
\end{equation}
where $\phi_{r,t}$ is the local similarity measure introduced in
Section \ref{sec:Measuring-Similarity-for}, $P_{n}=\frac{1}{n}\iota_{n}\iota_{n}^{\prime}$
and $P_{n}^{\perp}=I_{n}-P_{n}$ are the orthogonal projection matrices,
$\iota_{n}$ is the $n$-dimensional vector of ones, and $I_{n}$
is the n-dimensional identity matrix. The structure of this recursive
equation is analogous to the classical GARCH equation for volatility.
The autoregressive term with $\beta$ accounts for persistence in
correlation dynamics, while the term with $\phi_{r,t-1}$ represents
an observation-driven innovation\footnote{The classification of time-varying parameter models into the classes
of observation-driven and parameter-driven models goes back to \citet{Cox_1981}.
See also an instructive discussion about observation-driven modeling
in \citet{Koopman_Lucas_Scharth_2016}.} for the correlation dynamics allowing for an asymmetric effect that
depends on whether the prevailing direction of stock returns is positive
or negative. Since $z_{t}$ is elliptical with homogeneous variances
and correlations, by Proposition \ref{prop:P4}, $\phi_{r,t}$ represents
a relevant signal of $\phi_{\rho,t}$ with a known constant bias term,
$\mathbb{E}(\phi_{r,t}|\mathcal{F}_{t-1})=\phi_{\rho,t}-\omega_{n}$.
Such fixed bias is harmless as it is absorbed by the constant coefficient
$\alpha$ in \eqref{eq:deco_garch}. At the same time, $\phi_{r,t}$
is robust to the presence of outliers and fat-tailed distributions
of the observed returns. This is a particularly useful feature in
modeling financial returns, where the fat tails and jumps are commonly
found among stylized empirical regularities. 

An apparent advantage of the specification in \eqref{eq:deco_garch}
is an unconstrained support of $\phi_{\rho,t}$ which allows to avoid
any additional restrictions on the dynamic process for ensuring positive
definiteness of $C_{t}$. Recall that, in the classical DCC-GARCH
models, the conditional correlation dynamics is specified through
the matrix process, such that the resulting matrix needs an extra
adjustment for positive definiteness and for the estimated conditional
correlations to remain valid (see also \citet{Aielli_2013}, \citet{Brownlees_Llorens_2023}).
We note that many extensions, which are common for the DCC-GARCH models,
can be readily incorporated into this unconstrained specification.
For example, these extensions may include correlation targeting or
incorporating realized measures of correlation in spirit of \citet{Archakov_Hansen_Lunde_2025}.

We emphasize that, unlike classical DCC-type recursions built on outer
products of standardized returns, the correlation innovation in \eqref{eq:deco_garch}
is well defined for every elliptical generator, including infinite-variance
families such as the symmetric stable or Cauchy, since it enters only
through the direction of $z_{t}$. In this respect the proposed specification
is robust along a dimension that the DECO benchmark, which presumes
finite second moments, does not share.

We note that for $n=2$ the model reduces to a simple bivariate DCC
model, where $\phi_{\rho,t}$ represents a Fisher transform of the
conditional correlation coefficient, while the empirical resemblance
measure $\phi_{r,t}$ is given by \eqref{eq:phi_r}. For the bivariate
case, the model becomes closely related to the angular DCC model in
\citet{Jarjour_Chan_2020}, where the conditional correlation dynamics
was specified for non-transformed $\rho_{t}$, and $r_{t}=\frac{2z_{1,t}z_{2,t}}{z_{1,t}^{2}+z_{2,t}^{2}}$
was used as a dynamic innovation term. Although $r_{t}$ is a robust
signal of correlation, it is not an unbiased measure of $\rho_{t}$
because, in general, $\mathbb{E}(r_{t}|\mathcal{F}_{t-1})\neq\rho_{t}$.
Furthermore, extra restrictions for the parameter coefficients must
be imposed to ensure positive definiteness. In light of this, the
specification in \eqref{eq:deco_garch} appears to be a more natural
and convenient approach to modeling dynamic correlations.

An important advantage of the DCC structure is that the model can
be effectively estimated in two steps. In the first step, the univariate
GARCH models are estimated individually for $i=1,...,n$, and the
standardized returns, $\hat{z}_{t}$, are obtained. Note that since
the correlation recursion is invariant to the common normalization
of $\hat{z}_{t}$, the first-step filters may be estimated by Gaussian
QML without affecting the second step. In the second step, the conditional
correlation dynamics of $C_{t}$, given in \eqref{eq:deco_garch},
is estimated by using $\hat{z}_{t}$ obtained in the first step. Such
two-stage estimation can facilitate the computation complexity dramatically,
especially if high dimensional return vectors are modeled.

The equicorrelation assumption imposes a tight constraint on $C_{t}$,
which is rarely empirically plausible. This structure, however, provides
a substantial dimension reduction since, instead of modeling $\frac{n(n-1)}{2}$
dynamic correlations, we effectively model only a single common correlation
coefficient. The estimated dynamics can be interpreted as a time-varying
average correlation level, or an index of market co-movement, and
can serve as a useful state variable in various contexts. For example,
when a sufficiently representative sample of assets is available,
$\phi_{\rho,t}$ can be used as a barometer of diversification potential
for portfolio management purposes, or as an aggregate correlation
index which can be informative about overall market uncertainty and
risk.

\begin{figure}[ph]
\begin{centering}
\includegraphics[scale=0.65]{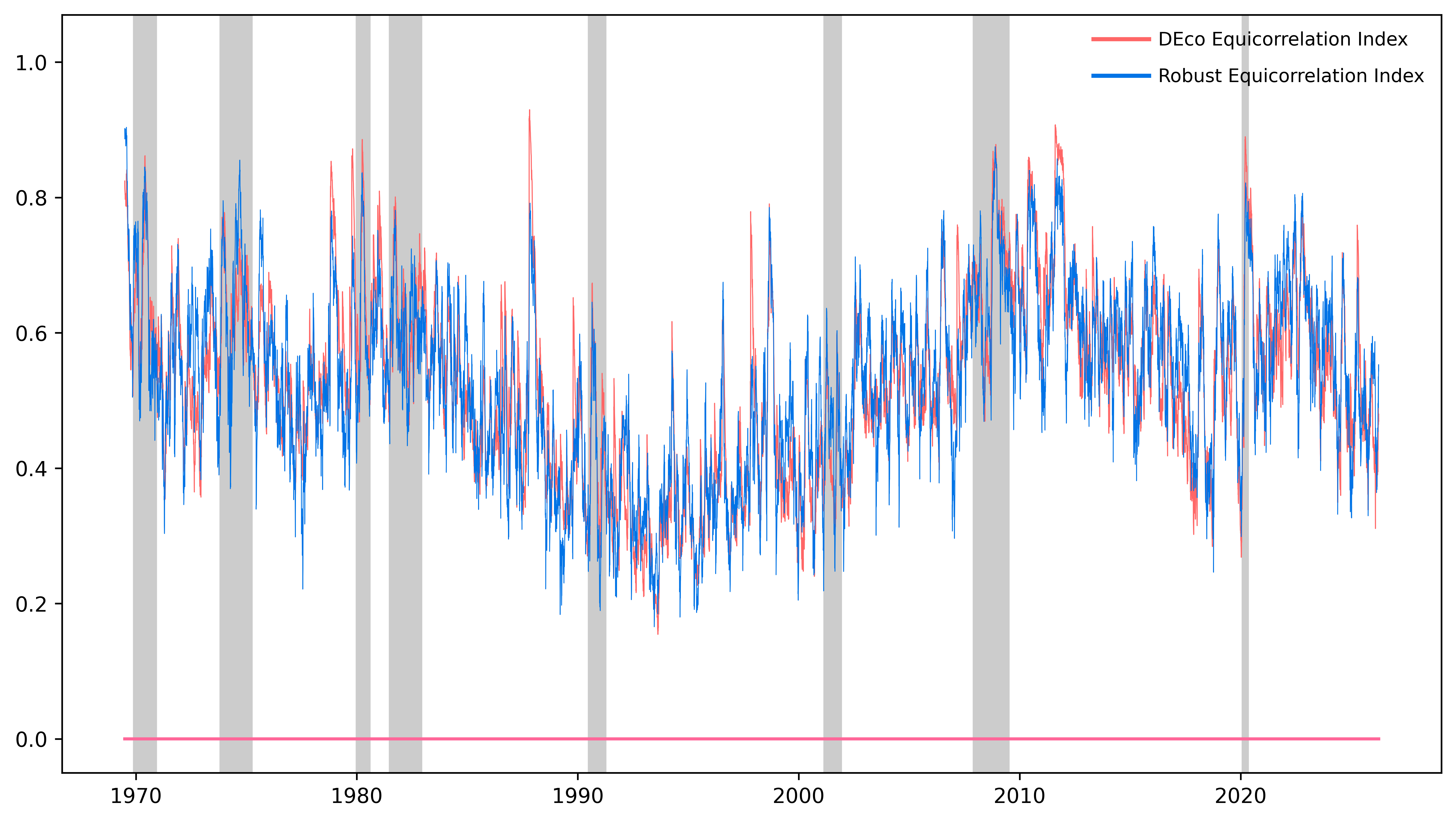}
\par\end{centering}
\caption{\footnotesize\label{fig:rmg_index} Equicorrelation Multivariate
GARCH models estimated with daily returns on 49 U.S. industry portfolios
from the Kenneth French data library. The conditional equicorrelation
dynamics estimated with the new robust specification (blue line) and
the equicorrelation dynamics estimated with the standard DECO model
(red line) are shown for the period between July 1969 and March 2026
($14,309$ trading days).}
\end{figure}

For illustration purposes, we estimate the new model with daily close-to-close
value-weighted returns on $49$ industry portfolios from the Kenneth
French Data Library. The estimation period spans July 1969 to March
2026, comprising a total of 14,309 trading days. The results are illustrated
in Figure \ref{fig:rmg_index}. The conditional equicorrelation index
estimated with the new robust model is shown with a blue line, while
the index estimated with a standard DECO model by \citet{Engle_Kelly_2012}
is shown in red. We observe that both indices exhibit very similar
low-frequency cyclical trajectory of the common correlation level
which globally ranges from about $0.2$ to $0.8$. Interestingly,
the figure suggests that the economic downturns widely considered
as global, such as recession of 1973-75, early 1980s recession, and
the Great Recession in 2007-09 coupled with the European Debt Crisis
2010-12, appear to be the major turning points when the correlation
level shifts into a long-term contraction after reaching the peak
($\rho_{t}\approx0.8$). 

The discrepancy between two estimated correlations comes through rather
short-run dynamics, where the indices may demonstrate visibly different
magnitudes. For example, we observe that the new correlation index
exhibits a more modest response to fleeting or relatively short-lived
events such as the \textquotedbl Black Monday\textquotedbl{} (October
19, 1987), the \textquotedbl Mini-crash\textquotedbl{} (October 27,
1997), and the outbreak of COVID-19 (February-March, 2020), where
many stocks exhibited co-directional extreme returns over a local
period of time. This evidence points to distinctive informational
content in the new robust estimates of conditional correlations, and
its evaluation suggests a promising direction for further empirical
analysis. 

\section{Conclusion}

We consider an empirical measure of association that reflects the
statistical similarity between random variables and is grounded in
the ideas of \citet{Thorndike_1905} and \citet{Fisher_1919}. The
similarity is defined as the relative extent of joint variation along
the direction of the vector of ones, which can be viewed as the direction
of perfect similarity in both sign and magnitude. In this paper, we
adapt the similarity measure to the robust estimation of the correlation
coefficient and to the corresponding inference. We propose both sample-average
and maximum-likelihood versions of the estimator, which are intrinsically
insensitive to extreme observations arising from fat-tailed distributions
and outliers.

We analyze the statistical properties of the proposed class of estimators
for elliptical random vectors. In the bivariate case, when the variables
are scale homogeneous, the estimators are consistent for the Fisher
transformation of the correlation coefficient; the sample-average
version admits an exact finite-sample distribution that is invariant
over the elliptical class, while the maximum-likelihood version attains
the inverse Fisher information of the hyperbolic secant family. In
the more realistic setting where the variables may be heteroskedastic,
the variables must first be re-scaled before the estimator is applied.
We therefore develop a two-stage maximum-likelihood estimator, based
on a preliminary re-scaling by the median of the observed log-ratios,
and show that it retains the robustness and asymptotic properties
of the corresponding homogeneous-scale estimator.

In contrast to many popular robust correlation estimators, which primarily
stabilize point estimation, the proposed estimators also provide a
robust and invariant sampling distribution. Consequently, they can
be used to construct robust confidence intervals in the presence of
noisy and/or heavy-tailed data, to perform two-sample tests for the
equality of correlations, to develop CUSUM-type statistics for robust
real-time monitoring of correlation levels, and to support a wide
range of other robust inference applications. Another promising direction
is the development of a new robust regression framework, in which
$\tanh(\gamma)\cdot e^{\eta}$ provides a natural robust analogue
of the slope coefficient.

The intrinsic robustness of the similarity estimator comes at the
cost of reduced efficiency. This motivates the development of efficiency-improving
modifications based, for example, on subsampling methods. This research
direction is particularly promising in the context of intraday financial
data, where the large number of high-frequency observations can be
further exploited. Another potentially useful methodological contribution
lies in the development of new methods for the composite estimation
of large-scale correlation matrices. For instance, the individual
correlation coefficients could first be estimated separately using
the robust similarity estimator, after which the resulting matrix
could be projected onto the space of positive definite correlation
matrices according to an appropriate optimality criterion. We leave
these topics for future research.

The similarity measure extends to multiple (more than two) variables,
where it can be interpreted as a robust measure of the overall correlation
level within the system. We analyze a tractable special case in which
scales and correlations are homogeneous, and show that the robust
similarity estimator naturally generalizes to a consistent estimator
of the common coefficient of the log-transformed correlation matrix,
with an exact and invariant sampling distribution. This result establishes
an interesting connection with the Generalized Fisher Transformation
(GFT) for correlation matrices introduced in \citet{Archakov_Hansen_2021}
and its application to block-regularized correlation structures; see
\citet{Archakov_Hansen_2024}.

We illustrate the empirical performance of the robust similarity estimator
using intraday stock returns. We perform robust interval estimation
for more than $800,000$ pairwise daily correlations, each at thirteen
sampling frequencies. The resulting estimates are stable, persistent,
and in good agreement with the point correlation estimates obtained
with a popular robust alternative, the Kendall estimator, at all but
the highest sampling frequencies. These findings suggest that the
similarity estimator provides a reliable tool for the estimation of,
and inference on, financial correlations using high-frequency data.
It may be particularly useful for analyzing especially noisy asset
classes, such as cryptocurrencies.

As an econometric application, we develop a novel robust multivariate
GARCH model in which the conditional correlation process is modeled
under the matrix logarithm transformation, which naturally accommodates
the similarity measure. The proposed specification preserves the positive
definiteness of the filtered correlation matrix while ensuring robustness
in the presence of fat tails and outliers. In financial econometric
applications, the model may be useful for the robust estimation of
market betas (in the bivariate case) and for the regularized estimation
of high-dimensional dynamic correlation matrices (in the multivariate
case). A straightforward extension would be to incorporate robust
realized measures of correlation, computed with the similarity estimator
from intraday data, in the spirit of \citet{Archakov_Hansen_Lunde_2025}.
Furthermore, the similarity estimator may also be suitable for modeling
dynamic correlations within the score-driven framework proposed by
\citet{CKL13}, or in parameter-driven models such as state-space
or stochastic volatility models. We leave these avenues for future
research.

\bibliographystyle{apalike}
\bibliography{literature}

\noindent{\small\newpage}{\small\par}

\appendix

\section*{Appendix }

\setcounter{equation}{0}\renewcommand{\theequation}{A.\arabic{equation}}
\setcounter{lem}{0}\renewcommand{\thelem}{A.\arabic{lem}}
\begin{lem}
\label{lem:1} Let $x=(x_{1},x_{2},...,x_{n})^{\prime}$ is a $n$-variate
random vector which has an elliptical distribution with zero location
and positive-definite scatter matrix $\Sigma$, so its probability
density function can be expressed as $f_{x}(x)=g(x^{\prime}\Sigma^{-1}x)$,
for $x\in\mathbb{R}^{n}$, where $g(t)$ is some non-negative valued
function. Then, the following identity holds,
\[
\int_{0}^{+\infty}t^{\frac{n}{2}-1}g(t)dt=\frac{1}{\sqrt{\pi^{n}|\Sigma|}}\Gamma\Bigl(\frac{n}{2}\Bigl),
\]
where $|\Sigma|$ denotes the determinant of $\Sigma$.
\end{lem}
\begin{proof}
Introduce $z=\Sigma^{-\frac{1}{2}}x$ and note that the scatter matrix
of $z$ is the identity matrix of size $n\times n$. The change of
variables from $x=(x_{1},x_{2},...,x_{n})^{\prime}$ to $z=(z_{1},z_{2},...,z_{n})^{\prime}$
leads to the transformed probability density function, $f_{z}(z)=|\Sigma|^{\frac{1}{2}}g(z^{\prime}z)$,
which is spherically symmetric around the origin. Apply now the hyper-spherical
coordinate transformation,
\[
\begin{aligned}z_{1}= & s\cdot\cos\theta_{1},\\
z_{2}= & s\cdot\sin\theta_{1}\cos\theta_{2},\\
\vdots\\
z_{n-1}= & s\cdot\sin\theta_{1}...\sin\theta_{n-2}\cos\theta_{n-1},\\
z_{n}= & s\cdot\sin\theta_{1}...\sin\theta_{n-2}\sin\theta_{n-1},
\end{aligned}
\]
where $s\geq0$, $\theta_{k}\in[0,\pi)$, for $k=1,...,n-2$, and
$\theta_{n-1}\in[0,2\pi)$, and the following identity holds $dz_{1}dz_{2}...dz_{n}=\Lambda\cdot dsd\theta_{1}...d\theta_{n-2}d\theta_{n-1}$,
where $\Lambda=s^{n-1}\sin\theta_{1}\sin^{2}\theta_{2}...\sin^{n-2}\theta_{n-2}$. 

Considering the total integral over the whole support leads to the
following condition
\[
\begin{aligned} & \int_{-\infty}^{+\infty}...\int_{-\infty}^{+\infty}f_{z}(z)dz_{1}...dz_{n} & = & |\Sigma|^{\frac{1}{2}}\int_{-\infty}^{+\infty}...\int_{-\infty}^{+\infty}g(z^{\prime}z)dz_{1}...dz_{n}\\
 &  & = & |\Sigma|^{\frac{1}{2}}\int_{0}^{2\pi}\int_{0}^{\pi}...\int_{0}^{\pi}\int_{0}^{+\infty}g(s^{2})\cdot\Lambda\cdot dsd\theta_{1}...d\theta_{n-2}d\theta_{n-1}\\
 &  & = & \frac{1}{2}|\Sigma|^{\frac{1}{2}}\int_{0}^{2\pi}d\theta_{n-1}\int_{0}^{\pi}\sin^{n-2}\theta_{n-2}d\theta_{n-2}...\int_{0}^{\pi}\sin\theta_{1}d\theta_{1}\int_{0}^{+\infty}g(s^{2})\cdot s^{n-2}ds^{2}\\
 &  & = & \frac{1}{2}\mathcal{S}_{n-1}|\Sigma|^{\frac{1}{2}}\int_{0}^{+\infty}t^{\frac{n}{2}-1}g(t)dt,
\end{aligned}
\]
where we use identity $z^{\prime}z=z_{1}^{2}+z_{2}^{2}+...+z_{n}^{2}=s^{2}$,
and $|\mathcal{S}_{n-1}|=2\pi^{\frac{n}{2}}\Gamma^{-1}\Bigl(\frac{n}{2}\Bigl)$
is the surface area of $(n-1)$-dimensional unit sphere, $\mathcal{S}_{n-1}$.
The condition stated in the lemma is obtained by equating the derived
expression to $1$.
\end{proof}
\noindent\textbf{Proof of Theorem} \ref{prop:T1}. The proof is closely
related to the similar derivation in \citet{Fisher_1919} and adds
a simple extension of the result for the general case of bivariate
elliptical distributions. 

Since $x$ has an elliptical distribution, its probability density
function can be written as $f_{x}(x)=g(x^{\prime}\Sigma^{-1}x)$,
where $x\in\mathbb{R}^{2}$ and $g(t)$ is a non-negative valued function.
\[
f_{x}(x)=g(x^{\prime}\Sigma^{-1}x)=g\Biggl(\frac{x_{1}^{2}+x_{2}^{2}-2\rho x_{1}x_{2}}{\sigma^{2}(1-\rho^{2})}\Biggl)
\]
Apply the polar coordinate transformation, $x_{1}=s\cdot\cos\theta$
and $x_{2}=s\cdot\sin\theta$, where $s\geq0$ and $\theta\in[0,2\pi)$,
and obtain the corresponding probability density function of the transformed
variables,
\[
f_{s,\theta}(s,\theta)=s\cdot g\Biggl(\frac{s^{2}}{\sigma^{2}}\cdot\frac{(1-\rho\sin2\theta)}{1-\rho^{2}}\Biggl).
\]
After the another variable change, $r=\sin2\theta$, for $\theta\in[0,\frac{\pi}{2})$,
we have
\[
f_{s,r}(s,r)=\frac{2s}{\sqrt{1-r^{2}}}\cdot g\Biggl(\frac{s^{2}}{\sigma^{2}}\cdot\frac{1-\rho r}{1-\rho^{2}}\Biggl),
\]
where we use $\frac{\partial\theta}{\partial r}=\frac{1}{2\sqrt{1-r^{2}}}$
and multiply the expression by $4$ to account for all four elementary
sectors of the polar coordinate system, $\theta\in[0,2\pi)$. We obtain
the marginal density for $r$ by integrating with respect to $s$,
\[
\begin{aligned} & f_{r}(r)=\int_{0}^{+\infty}f_{s,r}(s,r)ds & = & \frac{2}{\sqrt{1-r^{2}}}\cdot\int_{0}^{+\infty}s\cdot g\Biggl(\frac{s^{2}}{\sigma^{2}}\cdot\frac{1-\rho r}{1-\rho^{2}}\Biggl)ds\\
 &  & = & \frac{\sigma^{2}(1-\rho^{2})}{\sqrt{1-r^{2}}(1-\rho r)}\cdot\int_{0}^{+\infty}g(t)dt=\frac{1}{\pi(1-\rho r)}\cdot\sqrt{\frac{1-\rho^{2}}{1-r^{2}}}
\end{aligned}
\]
where we used the result from Lemma \ref{lem:1} for the special case
$n=2$, and $|\Sigma|=\sigma^{4}(1-\rho^{2})$. Finally, after the
Fisher transformation, $\phi_{r}=\frac{1}{2}\log\Bigl(\frac{1+r}{1-r}\Bigl)=\text{atanh}(r)$
and $\phi_{\rho}=\frac{1}{2}\log\Bigl(\frac{1+\rho}{1-\rho}\Bigl)=\text{atanh}(\rho)$,
the probability density function reads
\[
\begin{aligned} & f_{\phi_{r}}(\phi_{r}) & = & \frac{1}{\pi}\cdot\frac{1-\tanh^{2}\phi_{r}}{1-\tanh\phi_{r}\cdot\tanh\phi_{\rho}}\cdot\sqrt{\frac{1-\tanh^{2}\phi_{\rho}}{1-\tanh^{2}\phi_{r}}}=\frac{1}{\pi}\cdot\frac{\sqrt{(1-\tanh^{2}\phi_{r})(1-\tanh^{2}\phi_{\rho})}}{1-\tanh\phi_{r}\cdot\tanh\phi_{\rho}}\\
 &  & = & \frac{1}{\pi}\cdot\frac{2}{e^{\phi_{r}-\phi_{\rho}}+e^{-(\phi_{r}-\phi_{\rho})}}=\frac{1}{\pi}\cdot\text{sech}(\phi_{r}-\phi_{\rho}),
\end{aligned}
\]
where we use $\frac{\partial r}{\partial\phi_{r}}=\frac{\partial}{\partial\phi_{r}}\tanh\phi_{r}=1-\tanh^{2}\phi_{r}$.
$\square$

\noindent\textbf{Proof of Proposition} \ref{prop:P1}. For some $\gamma$,
we write $\psi(x_{t};\gamma)=\tanh(\phi_{r,t}-\gamma)$, $\Psi_{T}(\gamma)=\frac{1}{T}\sum_{t=1}^{T}\psi_{t}(x_{t};\gamma)$,
and $\Psi(\gamma)=\mathbb{E}\psi(x_{t};\gamma)$. Note that $\hat{\gamma}_{ML}$
solves $\Psi_{T}(\gamma)=0$, as the first-order condition of the
strictly concave log-likelihood $\ell(\gamma)$.

\textit{{[}Consistency of $\hat{\gamma}_{ML}${]}} Since the angular
density of $x_{t}$ is bounded (Lemma \ref{lem:1}), $\mathbb{P}(r_{t}=\pm1)=0$,
so almost surely $|\phi_{r,t}|<\infty$ for all $t$. On this event,
$\Psi_{T}$ is continuous and strictly decreasing, with $\frac{\partial}{\partial\gamma}\Psi_{T}(\gamma)=-\frac{1}{T}\sum_{t=1}^{T}\text{sech}^{2}(\phi_{r,t}-\gamma)<0$
and $\Psi_{T}(\gamma)\to\mp1$ as $\gamma\to\pm\infty$; a unique
root exists by the Intermediate Value Theorem.

By Theorem \ref{prop:T1}, the variables $Z_{t}=\phi_{r,t}-\phi_{\rho}$
are i.i.d. with density $\pi^{-1}\text{sech}(\cdot)$, for every elliptical
generator. The population score $\Psi(\gamma)$ is strictly decreasing,
and $\Psi(\phi_{\rho})=\mathbb{E}\tanh(Z_{t})=0$ by the symmetry
of $Z_{t}$ and the oddness of $\tanh(\cdot)$; hence $\phi_{\rho}$
is its unique zero. Since $|\psi(\gamma)|\le1$ and $|\partial\psi/\partial\gamma|\le1$,
the class $\{\psi\}$ is uniformly bounded and Lipschitz, so the Uniform
LLN (\citet{NeweyMcFadden1994}, Lemma 2.4) gives $\sup_{\gamma\in K}|\Psi_{T}(\gamma)-\Psi(\gamma)|\overset{a.s.}{\rightarrow}0$
on compacts. Fix $\varepsilon>0$: then $\Psi(\phi_{\rho}-\varepsilon)>0>\Psi(\phi_{\rho}+\varepsilon)$,
so a.s. for all $T$ large, $\Psi_{T}(\phi_{\rho}-\varepsilon)>0>\Psi_{T}(\phi_{\rho}+\varepsilon)$,
and the strict monotonicity of $\Psi_{T}$ places its root in $(\phi_{\rho}-\varepsilon,\phi_{\rho}+\varepsilon)$.
As $\varepsilon$ was arbitrary, $\hat{\gamma}_{ML}\overset{a.s.}{\rightarrow}\phi_{\rho}$.

\textit{{[}Asymptotic distribution of $\hat{\gamma}_{ML}${]}} By
the mean value theorem, for some $\bar{\gamma}$ between $\hat{\gamma}_{ML}$
and $\phi_{\rho}$,
\[
0=\Psi_{T}(\phi_{\rho})-\Big[\frac{1}{T}\sum_{t=1}^{T}\text{sech}^{2}\big(\phi_{r,t}-\bar{\gamma}\big)\Big]\cdot\big(\hat{\gamma}_{ML}-\phi_{\rho}\big).
\]
By the Uniform LLN for the (bounded, Lipschitz) class $\{\text{sech}^{2}(\cdot)\}$,
continuity of its limit, and $\bar{\gamma}\overset{a.s.}{\rightarrow}\phi_{\rho}$,
the bracketed term converges a.s. to
\[
\mathbb{E}\text{sech}^{2}(Z_{t})=\frac{1}{\pi}\int_{-\infty}^{+\infty}\text{sech}^{3}dz=\frac{1}{2},
\]
The summands $\tanh(Z_{t})$ are i.i.d., bounded, have mean zero by
symmetry and oddness, and variance

\[
\mathbb{E}\tanh^{2}(Z_{t})=\frac{1}{\pi}\int_{-\infty}^{+\infty}\text{sech}z(1-\text{sech}^{2}z)dz=\frac{1}{\pi}\Bigl(\pi-\frac{\pi}{2}\Bigl)=\frac{1}{2},
\]
using standard integration results $\int_{-\infty}^{+\infty}\text{sech}zdz=\pi$
and $\int_{-\infty}^{+\infty}\text{sech}^{3}zdz=\frac{\pi}{2}$; see
\citet{gradshteyn1996table}. By the Lindeberg-L\'{e}vy CLT and the
Slutsky's Theorem, we obtain
\[
\sqrt{T}(\hat{\gamma}_{ML}-\phi_{\rho})=\Bigl[\frac{1}{2}+o_{p}(1)\Bigl]^{-1}\sqrt{T}\Psi_{T}(\phi_{\rho})\overset{d}{\rightarrow}\mathcal{N}(0,2).
\]
The equality $\mathbb{E}\tanh^{2}(Z_{t})=\mathbb{E}\text{sech}^{2}(Z_{t})=\tfrac{1}{2}=\mathcal{I}(\phi_{\rho})$
is the information equality: the hyperbolic secant location model
is correct for every elliptical generator, so the limit variance equals
$\mathcal{I}^{-1}(\phi_{\rho})=2$. $\square$

\noindent\textbf{Proof of Proposition} \ref{prop:P2}. We start with
introducing the orthogonal vectors $\iota_{+}=(1,1)^{\prime}$ and
$\iota_{-}=(1,-1)^{\prime}$, and note that $\phi_{r}=\frac{1}{2}\log\Bigl(\frac{1+r}{1-r}\Bigl)=\frac{1}{2}\log(\iota_{+}^{\prime}x)^{2}-\frac{1}{2}\log(\iota_{-}^{\prime}x)^{2}$.
We can use the following stochastic representation for the elliptical
vector $x=(x_{1},x_{2})^{\prime}$,
\[
x\overset{d}{=}R\cdot A\cdot u,
\]
where $R$ is a non-negative random variable, $u=(\cos\theta,\sin\theta)^{\prime}$
with $\theta$ is independent of $R$ and uniform on $[0,2\pi)$,
so, $u$ is distributed uniformly on $\mathcal{S}_{1}$, and $A$
is such that
\[
A=\left(\begin{array}{cc}
\sigma_{1} & 0\\
\frac{\sigma_{12}}{\sigma_{1}} & \sigma_{2}\sqrt{1-\bigl(\frac{\sigma_{12}}{\sigma_{1}\sigma_{2}}\bigl)^{2}}
\end{array}\right)=\left(\begin{array}{cc}
\sigma_{1} & 0\\
\sigma_{2}\rho & \sigma_{2}\sqrt{1-\rho^{2}}
\end{array}\right),\qquad\text{and}\qquad AA^{\prime}=\Sigma,
\]
where $\Sigma$ is a positive definite scatter matrix of $x$, as
given in \eqref{eq:Sigma_2x2_hetero}.

For the first term of $\phi_{r}$, we have 
\[
\begin{aligned} & \log(\iota_{+}^{\prime}x)^{2} & = & \log R^{2}+\log(\iota_{+}^{\prime}Au)^{2}\\
 &  & = & \log R^{2}+\log\Biggl(\Bigl(\sigma_{1}+\sigma_{2}\rho\Bigl)\cos\theta+\sigma_{2}\sqrt{1-\rho^{2}}\sin\theta\Biggl)\\
 &  & = & \log R^{2}+\log||\iota_{+}^{\prime}A||^{2}+\log\cos^{2}(\theta-\varphi_{1}),
\end{aligned}
\]
where $\varphi_{1}$ is such that $\cos\varphi_{1}=||\iota_{+}^{\prime}A||^{-1}\Bigl(\sigma_{1}+\sigma_{2}\rho\Bigl)$
and $\sin\varphi_{1}=||\iota_{+}^{\prime}A||^{-1}\sigma_{2}\sqrt{1-\rho^{2}}$.
Similarly, we obtain
\[
\begin{aligned} & \log(\iota_{-}^{\prime}x)^{2} & = & \log R^{2}+\log(\iota_{-}^{\prime}Au)^{2}\\
 &  & = & \log R^{2}+\log\Biggl(\Bigl(\sigma_{1}-\sigma_{2}\rho\Bigl)\cos\theta-\sigma_{2}\sqrt{1-\rho^{2}}\sin\theta\Biggl)\\
 &  & = & \log R^{2}+\log||\iota_{-}^{\prime}A||^{2}+\log\cos^{2}(\theta-\varphi_{2}),
\end{aligned}
\]
where $\varphi_{2}$ is such that $\cos\varphi_{2}=||\iota_{-}^{\prime}A||^{-1}\Bigl(\sigma_{1}-\sigma_{2}\rho\Bigl)$
and $\sin\varphi_{2}=-||\iota_{-}^{\prime}A||^{-1}\sigma_{2}\sqrt{1-\rho^{2}}$.
Then, the similarity measure can be written as follows,
\[
\phi_{r}=\frac{1}{2}\log\frac{||\iota_{+}^{\prime}A||^{2}}{||\iota_{-}^{\prime}A||^{2}}+D(\theta),
\]
where we denote $D(\theta)=\frac{1}{2}\Bigl(\log\cos^{2}(\theta-\varphi_{1})-\log\cos^{2}(\theta-\varphi_{2})\Bigl)$.
Consider the reflection $\mathcal{T}\theta=\varphi_{1}+\varphi_{2}-\theta$
($\text{mod}$$2\pi$), and note that the distribution of $\theta$,
and of $D(\theta)$, is preserved due to an isometry of the circle,
$\mathcal{T}\theta\overset{d}{=}\theta$ and $D(\mathcal{T}\theta)\overset{d}{=}D(\theta)$.
Since $\cos^{2}(\mathcal{T}\theta-\varphi_{1})=\cos^{2}(\theta-\varphi_{2})$
and $\cos^{2}(\mathcal{T}\theta-\varphi_{2})=\cos^{2}(\theta-\varphi_{1})$,
we find that $D(\mathcal{T}\theta)=-D(\theta)$, and therefore $D(\theta)\overset{d}{=}-D(\theta)$.
As a consequence, $D(\theta)$ has a symmetric distribution around
zero, and, therefore, $\phi_{r}$ is distributed symmetrically around
its mean,
\[
\mathbb{E}(\phi_{r})=\frac{1}{2}\log\frac{||\iota_{+}^{\prime}A||^{2}}{||\iota_{-}^{\prime}A||^{2}}=\frac{1}{2}\log\frac{(1+\xi)}{(1-\xi)}=\phi_{\xi},
\]
where we denote $\xi=\frac{2\sigma_{12}}{\sigma_{1}^{2}+\sigma_{2}^{2}}$.

Consider $\phi_{r}^{2}=\mathbb{E}(\phi_{r})^{2}+2\mathbb{E}(\phi_{r})\cdot D(\theta)+D(\theta)^{2}$,
and using that
\[
\mathbb{E}\log\cos^{2}(\theta-\varphi)=\frac{1}{2\pi}\int_{0}^{2\pi}\log\cos^{2}(\theta-\varphi)d\theta=\frac{2}{\pi}\int_{0}^{\pi}\log|\cos\theta|d\theta=-2\log2
\]
and
\[
\mathbb{E}\log^{2}\cos^{2}(\theta-\varphi)=\frac{1}{2\pi}\int_{0}^{2\pi}\log^{2}\cos^{2}(\theta-\varphi)d\theta=\frac{4}{\pi}\int_{0}^{\pi}\log^{2}|\cos\theta|d\theta=4\Bigl(\log^{2}2+\frac{\pi^{2}}{12}\Bigl),
\]
which are the standard integral result (see Section 4.22 in \citet{gradshteyn1996table}
for reference), we can obtain the variance of $\phi_{r}$,
\[
\begin{aligned} & V_{\phi_{r}} & = & \frac{1}{4}\mathbb{E}\log^{2}\cos^{2}(\theta-\varphi_{1})-\frac{1}{2}\mathbb{E}\log\cos^{2}(\theta-\varphi_{1})\log\cos^{2}(\theta-\varphi_{2})+\frac{1}{4}\mathbb{E}\log^{2}\cos^{2}(\theta-\varphi_{2})\\
 &  & = & 2\log^{2}2+\frac{\pi^{2}}{6}-\frac{1}{4\pi}\int_{0}^{2\pi}\log\cos^{2}(\theta-\varphi_{1})\log\cos^{2}(\theta-\varphi_{2})d\theta\\
 &  & = & \frac{\pi^{2}}{6}-\sum_{k=1}^{\infty}\frac{\cos2k(\varphi_{2}-\varphi_{1})}{k^{2}},
\end{aligned}
\]
where the last equality uses the Fourier expansion $\log|\cos\theta|=-\log2-\sum_{k=1}^{\infty}(-1)^{k}\frac{\cos(2k\theta)}{k}$
together with the orthogonality of $\cos2k\theta$ on $[0,2\pi)$,
so that the cross term equals $4\log^{2}2+2\sum_{k=1}^{\infty}\frac{2k(\varphi_{2}-\varphi_{1})}{k^{2}}$.
$\square$
\begin{lem}
\label{lem:2} Let $x=(x_{1},x_{2})^{\prime}$ is a bivariate elliptical
random vector with zero location and positive-definite scatter matrix
$\Sigma$, and let $w=\log|x_{2}|-\log|x_{1}|$, which is well defined
almost surely. Then $w$ is symmetrically distributed about $\eta_{c}=\log\frac{\sigma_{2}}{\sigma_{1}}$,
and has the density
\[
f_{W}(w)=\frac{1}{\pi}\cdot\frac{\sqrt{1-\rho^{2}}\cosh(w-\eta_{c})}{\cosh^{2}(w-\eta_{c})-\rho^{2}},
\]
for $w\in\mathbb{R}$, such that $\mathbb{E}(w)=\eta_{c}$ (the expectation
exists, since $f_{w}$ has exponentially decaying tails). The density
function is strictly decreasing in $|w-\eta_{c}|$, and $f_{W}(\eta_{c})=\frac{1}{\pi\sqrt{1-\rho^{2}}}>0$.
\end{lem}
\begin{proof}
Define the transformation 
\[
y=Bx,\qquad B=\left(\begin{array}{cc}
0 & \frac{\sigma_{2}}{\sigma_{1}}\\
\frac{\sigma_{1}}{\sigma_{2}} & 0
\end{array}\right),\qquad\text{such that}\quad y_{1}=\frac{\sigma_{2}}{\sigma_{1}}x_{2}\quad\text{and}\quad y_{2}=\frac{\sigma_{1}}{\sigma_{2}}x_{1}.
\]
Note that because of a linearity of the transformation $y$ is also
elliptical zero location vector, and, due to $B\Sigma B'=\Sigma$,
it has the same distribution as $x$. Therefore, $w(x)\overset{d}{=}w(y)$.
Since $w(y)=\log|y_{2}|-\log|y_{1}|=\log|x_{1}|-\log|x_{2}|+\log\frac{\sigma_{2}}{\sigma_{1}}=-w(x)+\eta_{c}$,
we have that $W\overset{d}{=}\eta_{c}-W$, i.e. $W$ is symmetrically
distributed around $\eta_{c}$.

Using the same stochastic representation for $x=(x_{1},x_{2})^{\prime}$
as in the proof for Proposition \ref{prop:P2}, $x\overset{d}{=}R\cdot A\cdot u$,
and representing the uniform angular component as $u=\frac{z}{||z||}$,
where $z\sim N(0,I_{2})$, we obtain 
\[
\frac{x_{2}}{x_{1}}=\frac{[Az]_{2}}{[Az]_{1}}=\frac{\sigma_{2}}{\sigma_{1}}\Bigl(\rho+\sqrt{1-\rho^{2}}\frac{z_{2}}{z_{1}}\Bigl),
\]
so the radial component $R$ is canceled out. Since $\frac{z_{2}}{z_{1}}\sim\text{Cauchy}(0,1)$,
and $\frac{\sigma_{2}}{\sigma_{1}}=e^{\eta_{c}}$, the density of
$h=\frac{x_{2}}{x_{1}}$ is derived as 
\[
f_{H}(h)=\frac{e^{\eta_{c}}\sqrt{1-\rho^{2}}}{\pi\Bigl[(h-e^{\eta_{c}}\rho)^{2}+e^{2\eta_{c}}(1-\rho^{2})\Bigl]},
\]
where $t\in\mathbb{R}$. Consider $W=\log|H|$, and note that the
inverse function consists of two branches, $H=e^{W}$ and $H=-e^{W}$.
Since $\Bigl|\frac{\partial h}{\partial w}\Bigl|=e^{w}$ for both
branches, we obtain
\[
\begin{aligned} & f_{W}(w) & = & e^{w}[f_{H}(e^{w})+f_{H}(-e^{w})]\\
 &  & = & \frac{\mu(w)\sqrt{1-\rho^{2}}}{\pi}\Bigl[\frac{1}{\mu(w)^{2}-2\rho\mu(w)+1}+\frac{1}{\mu^{2}+2\rho\mu(w)+1}\Bigl]\\
 &  & = & \frac{2\mu(w)\sqrt{1-\rho^{2}}(\mu(w)^{2}+1)}{\pi\Bigl[(\mu(w)^{2}+1)^{2}-4\rho^{2}\mu(w)^{2}\Bigl]},
\end{aligned}
\]
where we denote $\mu(w)=e^{w-\eta_{c}}$. Recall that $\cosh z=\frac{1}{2}\Bigl(e^{z}+\frac{1}{e^{z}}\Bigl)$,
and the expression for $f_{W}(w)$ stated in the Lemma follows.

Since $\text{sech}(w-\eta_{c})$ is strictly decreasing in $|w-\eta_{c}|$,
and due to $f_{W}(w)\sim\frac{\text{sech}(w-\eta_{c})}{1-\rho^{2}\text{sech}^{2}(w-\eta_{c})}$,
which is strictly increasing in $\text{sech}(w-\eta_{c})$, we have
that $f_{W}(w)$ is strictly decreasing in $|w-\eta_{c}|$. Therefore,
$f_{W}(w)$ attains its maximum at $w=\eta_{c}$ with $f_{W}(\eta_{c})=\frac{1}{\pi\sqrt{1-\rho^{2}}}$. 
\end{proof}
\noindent\textbf{Proof of Proposition} \ref{prop:P3}. We split our
proof into two main parts. The first part is related to the properties
of $\hat{\eta}$, while the second part is about $\hat{\gamma}_{ML*}$.

\textbf{\textit{The first stage}}\textbf{.} Estimator $\hat{\eta}=\underset{t=1,...,T}{\text{median}}(w_{t})$
is defined equivalently as a solution to optimization problem $\min_{\eta}\sum_{t=1}^{T}|w_{t}-\eta|$.
The problem is convex and coercive, so a minimizer exists. The set
of solutions represents a set of medians of $\{w_{t}\}$. Since $w_{t}$,
$t=1,...,T$, are i.i.d. with a continuous density, all $w_{t}$ are
distinct. For $T$ odd the minimizer is unique almost surely and equals
the order statistic $w_{(\lceil T/2\rceil)}$; for $T$ even the set
of minimizers is the interval $[w_{(T/2)},w_{(T/2+1)}]$, and we adopt
the conventional midpoint, which almost surely coincides with no observation.
All subsequent results hold for any measurable selection from the
minimizing set.

Since the distribution of $w_{t}$ is symmetric about $\eta_{c}$
and strictly positive in a neighborhood of $\eta_{c}$, as shown in
Lemma \ref{lem:2}, $\eta_{c}$ is the unique population median. Therefore,
the sample median $\hat{\eta}$ is strongly consistent for its unique
population counterpart, $\eta_{c}$; see, for example, \citet{serfling1980approximation}.
Also, the classical sample median CLT holds,
\[
\sqrt{T}(\hat{\eta}-\eta_{c})\overset{d}{\rightarrow}\mathcal{N}\Bigl(0,\frac{1}{4}f_{W}^{-2}(\eta_{c})\Bigl)=\mathcal{N}\Bigl(0,\frac{\pi^{2}}{4}(1-\rho^{2})\Bigl),
\]
and $\sqrt{T}(\hat{\eta}-\eta_{c})=O_{p}(1)$.

\textbf{\textit{The second stage}}\textbf{.} Only two properties of
$\hat{\eta}$ are used below: $\hat{\eta}\xrightarrow{a.s.}\eta_{c}$
and $\sqrt{T}(\hat{\eta}-\eta_{c})=O_{p}(1)$; the results therefore
hold for any first-stage estimator with these properties, for instance
one based on robust scale estimates. 

For some $\eta$ and $\gamma$, we denote $\psi(x_{t};\eta,\gamma)=\tanh(\tilde{\phi}_{r,t}(\eta)-\gamma)$,
$\Psi_{T}(\gamma;\eta)=\frac{1}{T}\sum_{t=1}^{T}\psi(x_{t};\eta,\gamma)$,
and $\Psi(\gamma;\eta)=\mathbb{E}\psi(x_{t};\eta,\gamma)$.

Note that even at observations for which $\tilde{r}_{t}(\eta)=\pm1$
and $\tilde{\phi}_{r,t}(\eta)$ diverges, $\psi(x_{t};\eta,\gamma)\in[-1,1]$
remains well defined and bounded. Also note that $\psi(x_{t};\eta,\gamma)$
has bounded derivatives, $\Bigl|\frac{\partial\psi}{\partial\gamma}\Bigl|=\text{sech}^{2}(\tilde{\phi}_{r}-\gamma)\leq1$
and 
\[
\begin{aligned} & \Bigl|\frac{\partial\psi}{\partial\eta}\Bigl| & = & \text{sech}^{2}(\tilde{\phi}_{r}-\gamma)\cdot\Bigl|\frac{\partial\tilde{\phi}_{r}}{\partial\eta}\Bigl|=\text{sech}^{2}(\tilde{\phi}_{r}-\gamma)\cdot\Biggl|\frac{2\tilde{x}_{1}(\eta)\tilde{x}_{2}(\eta)}{\tilde{x}_{1}^{2}(\eta)-\tilde{x}_{2}^{2}(\eta)}\Biggl|\\
 &  & = & \text{sech}^{2}(\tilde{\phi}_{r}-\gamma)\cdot\Biggl|\frac{\tilde{r}(\eta)}{\bigl|\sqrt{1-\tilde{r}^{2}(\eta)}\bigl|}\Biggl|=\text{sech}^{2}(\tilde{\phi}_{r}-\gamma)\cdot|\sinh(\tilde{\phi}_{r})|\leq2e^{2|\gamma|}e^{-|\tilde{\phi}_{r}|},
\end{aligned}
\]
where the last inequality uses $\text{sech}(z)\leq2e^{-|z|}$, $|\sinh(z)|\leq\frac{1}{2}e^{|z|}$
and $|\tilde{\phi}_{r}-\gamma|\geq|\tilde{\phi}_{r}|-|\gamma|$. The
latter bound vanishes as $|\tilde{\phi}_{r,t}|\rightarrow\infty$,
so $\eta\rightarrow\psi(x_{t};\eta,\gamma)$ is continuously differentiable
across the isolated points at which $\tilde{r}_{t}(\eta)=\pm1$. Consequently,
the classes $\{\psi\}$, $\{\frac{\partial\psi}{\partial\gamma}\}$,
$\{\frac{\partial\psi}{\partial\eta}\}$ are uniformly bounded and
Lipschitz continuous on a compact $(\eta,\gamma)$-set. Thus the Uniform
LLN can be applied to each (see \citet{NeweyMcFadden1994}, Lemma
2.4), with continuous limits and with expectation and differentiation
interchangeable by dominated convergence.

\textit{{[}Consistency of $\hat{\gamma}_{ML*}${]}} We begin by showing
that $\Psi_{T}(\gamma;\hat{\eta})=0$ has a unique root. If $T$ is
odd, $\hat{\eta}$ coincides with exactly one observation, the order
statistic $w_{(\lceil T/2\rceil)}$, due to $w_{t}$ are all distinct
almost surely for $t=1,...,T$. For this observation, $|\tilde{r}_{t}(\hat{\eta})|=1$,
and it contributes $\mp T^{-1}$ to the summation in $\Psi_{T}(\gamma;\hat{\eta})$.
If $T\geq3$, there are at least two observations for which $|\tilde{r}_{t}(\hat{\eta})|<1$,
so
\[
\frac{\partial}{\partial\gamma}\Psi_{T}(\gamma;\hat{\eta})=-\frac{1}{T}\sum_{t=1}^{T}\text{sech}^{2}(\tilde{\phi}_{r,t}(\hat{\eta})-\gamma)<0,
\]
for any $\gamma$, where the degenerate summand is constant in $\gamma$
and bounded by one in absolute value, while every non-degenerate summand
tends to $\mp1$ as $\gamma\to\pm\infty$; hence $\lim\sup_{\gamma\to+\infty}\Psi_{T}\le-(T-2)/T<0$
and $\lim\inf_{\gamma\to-\infty}\Psi_{T}\ge(T-2)/T>0$ for $T\ge3$,
and the Intermediate Value Theorem applies.

For some fixed $\eta\in\mathbb{R}$, the re-scaled vector $\tilde{x}_{t}(\eta)$
has unchanged correlation $\rho$ and scale ratio $\tilde{c}(\eta)=e^{\eta}\frac{\sigma_{1}}{\sigma_{2}}$.
Denote $\xi(\eta):=\xi(\tilde{c}(\eta),\rho)$. By Proposition \ref{prop:P2},
we have $\tilde{\phi}_{r,t}(\eta)=\phi_{\xi(\eta)}(\eta)+D(\eta)$,
where the last term is distributed symmetrically about zero. Due to
$\tanh(\cdot)$ is odd and bounded, $\Psi(\phi_{\xi}(\eta);\eta)=\mathbb{E}\text{tanh}D(\eta)=0$.
Together with $\frac{\partial\Psi}{\partial\gamma}=-\mathbb{E}\text{sech}^{2}(\tilde{\phi}_{r}(\eta)-\gamma)<0$,
it implies that $\gamma(\eta)=\phi_{\xi(\eta)}(\eta)$ is the unique
zero of $\Psi(\gamma;\eta)$ for every $\eta$, and for $\eta=\eta_{c}$
the re-scaled scales are equal, so $\xi=\rho$ and $\gamma(\eta_{c})=\phi_{\rho}$.
Note that for arbitrary $\varepsilon>0$, $\Psi(\phi_{\rho}-\varepsilon;\eta_{c})>0>\Psi(\phi_{\rho}+\varepsilon;\eta_{c})$.
By the Uniform LLN, continuity of $\Psi$, and $\hat{\eta}\overset{a.s.}{\rightarrow}\eta_{c}$,
it follows that $\Psi_{T}(\phi_{\rho}\pm\varepsilon;\hat{\eta})\overset{a.s.}{\rightarrow}\Psi(\phi_{\rho}\pm\varepsilon;\eta_{c})$.
As a result, $\Psi_{T}(\phi_{\rho}-\varepsilon;\hat{\eta})>0>\Psi_{T}(\phi_{\rho}+\varepsilon;\hat{\eta})$
holds almost surely for all sufficiently large $T$, and the strict
monotonicity of $\Psi_{T}(\gamma;\hat{\eta})$ places the root within
$(\phi_{\rho}-\varepsilon,\phi_{\rho}+\varepsilon)$. Since $\varepsilon>0$
was arbitrary, $\hat{\gamma}_{ML*}\overset{a.s.}{\rightarrow}\phi_{\rho}$.

\textit{{[}Asymptotic distribution of $\hat{\gamma}_{ML*}${]}} The
Mean Value Theorem gives
\[
0=\Psi_{T}(\phi_{\rho};\hat{\eta})-\Bigl[\frac{1}{T}\sum_{t=1}^{T}\text{sech}^{2}(\tilde{\phi}_{r,t}(\hat{\eta})-\bar{\gamma})\Bigl]\cdot(\hat{\gamma}_{ML*}-\phi_{\rho}),
\]
for some $\bar{\gamma}$ between $\hat{\gamma}_{ML*}$ and $\phi_{\rho}$.
Due to the Uniform LLN for the class $\{\text{sech}^{2}(\cdot)\}$,
continuity of its limit, and $(\bar{\gamma},\hat{\eta})\overset{a.s.}{\rightarrow}(\phi_{\rho},\eta_{c})$,
we obtain 
\[
\frac{1}{T}\sum_{t=1}^{T}\text{sech}^{2}(\tilde{\phi}_{r,t}(\hat{\eta})-\bar{\gamma})\overset{a.s.}{\rightarrow}\mathbb{E}\text{sech}^{2}(\tilde{\phi}_{r}(\eta_{c})-\phi_{\rho})=\frac{1}{\pi}\int_{-\infty}^{+\infty}\text{sech}^{3}zdz=\frac{1}{2},
\]
because $\tilde{\phi}_{r}(\eta_{c})$ has the same distribution $\pi^{-1}\text{sech}(\cdot)$
as $\phi_{r}$ in the scale-homogeneous case. Therefore,
\[
\sqrt{T}(\hat{\gamma}_{ML*}-\phi_{\rho})=\Bigl[\frac{1}{2}+o_{p}(1)\Bigl]^{-1}\sqrt{T}\Psi_{T}(\phi_{\rho};\hat{\eta}).
\]
Due to $\eta\rightarrow\Psi(\phi_{\rho};\eta)$ is continuously differentiable
(a.s.), for some $\bar{\eta}$ between $\eta_{c}$ and $\hat{\eta}$,
the Mean Value Theorem gives
\[
\sqrt{T}\Psi_{T}(\phi_{\rho};\hat{\eta})=\sqrt{T}\Psi_{T}(\phi_{\rho};\eta_{c})+M_{T}(\bar{\eta})\sqrt{T}(\hat{\eta}-\eta_{c}),
\]
where $M_{T}(\bar{\eta})=\frac{1}{T}\sum_{t=1}^{T}\frac{\partial\psi(x_{t};\eta,\phi_{\rho})}{\partial\eta}$.

We now prove that $M_{T}(\bar{\eta})=o_{p}(1)$. The Uniform LLN for
the class $\{\frac{\partial\psi}{\partial\eta}\}$ and the expectation-differentiation
interchange give $\sup_{\eta\in K}\Bigl|M_{T}(\eta)-\mathbb{E}\frac{\partial\psi}{\partial\eta}(x_{t};\eta,\phi_{\rho})\Bigl|\overset{a.s.}{\rightarrow}0$
on a compact set. Let us check the claim $\mathbb{E}\frac{\partial\psi}{\partial\eta}(x_{t};\eta_{c},\phi_{\rho})=0$.
Introduce notation $G(\delta,\eta)=\mathbb{E}\tanh(\delta+D(\eta))$
and $\delta(\eta)=\phi_{\xi(\eta)}(\eta)-\phi_{\rho}$, so $\mathbb{E}\psi(x_{t};\eta,\phi_{\rho})=G(\delta(\eta),\eta)$.
First, by oddness of $\tanh(\cdot)$ and symmetry of $D(\eta)$ about
zero shown in Proposition \ref{prop:P2}, $G(0,\eta)=0$ for any $\eta$,
hence $G_{\eta}(0,\eta)=0$ for any $\eta$. Second, we show that
$\delta(\eta_{c})=0$ and $\delta^{\prime}(\eta_{c})=0$. To prove
the last equality we use identity $\xi(\eta)=\frac{2\tilde{c}(\eta)\rho}{1+\tilde{c}^{2}(\eta)}$,
and then
\[
\delta^{\prime}(\eta)=\frac{\xi^{\prime}(\eta)}{1-\xi^{2}(\eta)},\qquad\text{and}\qquad\xi^{\prime}(\eta)=2\rho\frac{\tilde{c}(\eta)(1-\tilde{c}^{2}(\eta))}{(1+\tilde{c}^{2}(\eta))^{2}},
\]
so $\xi^{\prime}(\eta_{c})=0$ because $\tilde{c}(\eta_{c})=1$, and
$\delta^{\prime}(\eta_{c})=0$, the denominator being bounded away
from zero by $|\xi(\eta)|\leq|\rho|<1$. By the chain rule,

\[
\mathbb{E}\frac{\partial\psi}{\partial\eta}(x_{t};\eta_{c},\phi_{\rho})=G_{\delta}(0,\eta_{c})\cdot\delta^{\prime}(\eta_{c})+G_{\eta}(0,\eta_{c})=0,
\]
since $G_{\delta}(0,\eta_{c})=\mathbb{E}\text{sech}^{2}D(\eta_{c})\in[0,1]$
is a bounded term and $\delta^{\prime}(\eta_{c})=G_{\eta}(0,\eta_{c})=0$
as shown above. This result reflects the following intuition. The
first term vanishes because $\eta_{c}$ maximizes $|\xi(\tilde{c},\rho)|$
over the scale ratio (the movement of the \textit{center} of the distribution
of $\tilde{\phi}_{r}(\eta)$ has zero derivative at $\eta_{c}$),
and the second because, by the symmetry of $D(\eta)$ and the oddness
of $\tanh(\cdot)$, the deformation of the \textit{shape} of the distribution
contributes nothing at any $\eta$.

Since $\mathbb{E}\frac{\partial\psi}{\partial\eta}(x_{t};\eta,\phi_{\rho})$
is continuous and $\bar{\eta}\overset{a.s.}{\rightarrow}\eta_{c}$,
the desired result follows, $M_{T}(\bar{\eta})=o_{p}(1)$. Coupled
with $\sqrt{T}(\hat{\eta}-\eta_{c})=O_{p}(1)$ from the first part
of the proof, this result gives 
\[
\sqrt{T}\Psi_{T}(\phi_{\rho};\hat{\eta})=\sqrt{T}\Psi_{T}(\phi_{\rho};\eta_{c})+o_{p}(1).
\]

At $\eta=\eta_{c}$ the scales are homogeneous, so by Theorem \ref{prop:T1}
variables $Z_{t}=\tilde{\phi}_{r,t}(\eta_{c})-\phi_{\rho}$ are i.i.d.
with density $\pi^{-1}\text{sech}(\cdot)$, for every elliptical generator.
The summands $\tanh(Z_{t})$ are i.i.d., bounded, have mean zero by
symmetry and oddness, and variance, $\mathbb{E}\tanh^{2}(Z_{t})=\frac{1}{2}$
(see the proof for Proposition \ref{prop:P1}). By the Lindeberg-L\'{e}vy
CLT, $\sqrt{T}\Psi_{T}(\phi_{\rho};\eta_{c})\overset{d}{\rightarrow}\mathcal{N}(0,\frac{1}{2})$,
and, combined with the preceding result, $\sqrt{T}\Psi_{T}(\phi_{\rho};\hat{\eta})\overset{d}{\rightarrow}\mathcal{N}(0,\frac{1}{2})$
as well.

Finally, using the latter result and the Slutsky's Theorem, we obtain

\[
\sqrt{T}(\hat{\gamma}_{ML*}-\phi_{\rho})=\Bigl[\frac{1}{2}+o_{p}(1)\Bigl]^{-1}\sqrt{T}\Psi_{T}(\phi_{\rho};\hat{\eta})\overset{d}{\rightarrow}\mathcal{N}(0,2),
\]
which concludes the proof. $\square$

\noindent\textbf{Proof of Proposition} \ref{prop:P4}. A probability
density function of an elliptically distributed random vector $x\in\mathbb{R}^{n}$
has the form $f_{x}(x)=g(x^{\prime}\Sigma^{-1}x)$ for some non-negative
valued $g(t)$. Given the homogeneity assumption for scales and correlations,
the density can be written as
\[
f_{x}(x)=g\Biggl(x^{\prime}\Sigma^{-1}x\Biggl)=g\Biggl(x^{\prime}(\lambda_{+}^{-1}P_{n}+\lambda_{-}^{-1}P_{n}^{\perp})x\Biggl),
\]
where $\lambda_{+}=\sigma^{2}(1+(n-1)\rho)$ and $\lambda_{-}=\sigma^{2}(1-\rho)$
are eigenvalues of $\Sigma$ with multiplicities $1$ and $n-1$,
respectively, and orthogonal projection matrices $P_{n}$ and $P_{n}^{\perp}$
are defined in Section \ref{sec:Measuring-Similarity-for}.

We adapt the spherical reparametrization for the equicorrelation model.
For this, we introduce the overall radius, $s=||x||$, and the angle
between $x$ and the unit vector of ones, $\cos\theta=\frac{\iota_{n}^{\prime}x}{\sqrt{n}||x||}$,
such that $s\geq0$ and $\theta\in[0,\pi)$. Therefore, $x^{\prime}P_{n}x=s^{2}\cos^{2}\theta$
and $x^{\prime}P_{n}^{\perp}x=s^{2}\sin^{2}\theta$. The remaining
$n-2$ angular variables $\upsilon\in\mathcal{S}_{n-2}$ have the
uniform distribution on the $(n-2)$-dimensional unit sphere, which
lies in the orthogonal sub-space to vector $\iota_{n}$. By integrating
out these angular variables, we can obtain the probability density
function of $s$ and $\theta$ only,
\[
\begin{aligned} & f_{s,\theta}(s,\theta) & = & \int_{\upsilon\in\mathcal{S}_{n-2}}s^{n-1}\sin^{n-2}\theta\cdot g\Biggl(s^{2}\cdot\Bigl(\frac{\cos^{2}\theta}{\lambda_{+}}+\frac{\sin^{2}\theta}{\lambda_{-}}\Bigl)\Biggl)d\upsilon\\
 &  & = & |\mathcal{S}_{n-2}|\cdot s^{n-1}\sin^{n-2}\theta\cdot g\Biggl(s^{2}\cdot\Bigl(\frac{\cos^{2}\theta}{\lambda_{+}}+\frac{\sin^{2}\theta}{\lambda_{-}}\Bigl)\Biggl),
\end{aligned}
\]
where $|\mathcal{S}_{n-2}|=2\pi^{\frac{n-1}{2}}\Gamma^{-1}\Bigl(\frac{n-1}{2}\Bigl)$
is the surface area of $\mathcal{S}_{n-2}$.

Let us introduce a new variable $t=s^{2}\cdot\Bigl(\frac{\cos^{2}\theta}{\lambda_{+}}+\frac{\sin^{2}\theta}{\lambda_{-}}\Bigl)$
and obtain the marginal density of $\theta$ by integrating $f_{s,\theta}(s,\theta)$
with respect to $s$,
\[
\begin{aligned} & f_{\theta}(\theta)=\int_{0}^{+\infty}f_{s,\theta}(s,\theta)ds & = & \mathcal{S}_{n-2}\cdot\frac{\sin^{n-2}\theta}{2}\cdot\Bigl(\frac{\cos^{2}\theta}{\lambda_{+}}+\frac{\sin^{2}\theta}{\lambda_{-}}\Bigl)^{-1}\cdot\int_{0}^{+\infty}s^{n-2}g(t)dt\\
 &  & = & \mathcal{S}_{n-2}\cdot\frac{\sin^{n-2}\theta}{2}\cdot\Bigl(\frac{\cos^{2}\theta}{\lambda_{+}}+\frac{\sin^{2}\theta}{\lambda_{-}}\Bigl)^{-\frac{n}{2}}\cdot\int_{0}^{+\infty}t^{\frac{n}{2}-1}g(t)dt\\
 &  & = & \mathcal{B}^{-1}\Bigl(\frac{1}{2},\frac{n-1}{2}\Bigl)\cdot\lambda_{+}^{-\frac{1}{2}}\lambda_{-}^{-\frac{n-1}{2}}\cdot\sin^{n-2}\theta\cdot\Bigl(\frac{\cos^{2}\theta}{\lambda_{+}}+\frac{\sin^{2}\theta}{\lambda_{-}}\Bigl)^{-\frac{n}{2}}\\
 &  & = & \mathcal{B}^{-1}\Bigl(\frac{1}{2},\frac{n-1}{2}\Bigl)\cdot\Bigl(\frac{\lambda_{-}}{\lambda_{+}}\Bigl)^{\frac{1}{2}}\cdot\frac{1}{\sin^{2}\theta}\cdot\Bigl(\frac{\lambda_{-}}{\lambda_{+}}\cdot\frac{\cos^{2}\theta}{\sin^{2}\theta}+1\Bigl)^{-\frac{n}{2}}
\end{aligned}
\]
where, in the third line, we use the result from Lemma \ref{lem:1},
with $|\Sigma|=\lambda_{+}\lambda_{-}^{n-1}$, and the beta function
$\mathcal{B}\Bigl(\frac{1}{2},\frac{n-1}{2}\Bigl)$ arises from
\[
\frac{1}{2}\pi^{-\frac{n}{2}}\Gamma\Bigl(\frac{n}{2}\Bigl)\cdot\mathcal{S}_{n-2}=\pi^{-\frac{1}{2}}\Gamma^{-1}\Bigl(\frac{n-1}{2}\Bigl)\Gamma\Bigl(\frac{n}{2}\Bigl)=\mathcal{B}^{-1}\Bigl(\frac{1}{2},\frac{n-1}{2}\Bigl).
\]
After we apply the transformation, $\phi_{r}=\frac{1}{n}\log\Bigl(\cot^{2}\theta\Bigl)$,
for $\theta\in[0,\frac{\pi}{2})$, and denote $\phi_{\rho}=\frac{1}{n}\log\Bigl(\frac{\lambda_{+}}{\lambda_{-}}\Bigl)$,
the density is written as follows, 
\[
\begin{aligned} & f_{\phi_{r}}(\phi_{r}) & = & 2\mathcal{B}^{-1}\Bigl(\frac{1}{2},\frac{n-1}{2}\Bigl)\cdot\frac{n}{2}\frac{1}{1+e^{n\phi_{r}}}e^{\frac{1}{2}n\phi_{r}}\cdot e^{-\frac{1}{2}n\phi_{\rho}}\cdot(1+e^{n\phi_{r}})\cdot\Bigl(1+e^{-n\phi_{\rho}}\cdot e^{n\phi_{r}}\Bigl)^{-\frac{n}{2}}\\
 &  & = & \mathcal{B}^{-1}\Bigl(\frac{1}{2},\frac{n-1}{2}\Bigl)\cdot ne^{\frac{1}{2}n(\phi_{r}-\phi_{\rho})}\cdot\Bigl(1+e^{n(\phi_{r}-\phi_{\rho})}\Bigl)^{-\frac{n}{2}},
\end{aligned}
\]
where we use $\frac{\partial\theta}{\partial\phi_{r}}=-\frac{n}{2}\frac{1}{1+e^{n\phi_{r}}}e^{\frac{1}{2}n\phi_{r}}$.
Since we analyzed the transformation for $\theta\in[0,\frac{\pi}{2})$,
but $\theta\in[0,\pi)$, an additional factor $2$ is used in order
to incorporate both symmetric hemispheres.

For the transformed variable $\varphi_{n}=-n(\phi_{r}-\phi_{\rho})$,
the resulting density represents the Logistic-Beta distribution which
belongs to the class of Generalized Logistic distributions. For this
distribution, the moment generating function is available,
\[
M_{\varphi_{n}}(t)=\frac{\Gamma\Bigl(\frac{n-1}{2}+t\Bigl)\Gamma\Bigl(\frac{1}{2}-t\Bigl)}{\Gamma\Bigl(\frac{n-1}{2}\Bigl)\Gamma\Bigl(\frac{1}{2}\Bigl)},
\]
and the moments can be straightforwardly obtained. Therefore, for
the first moment of $\phi_{r}$ we have $\mathbb{E}(\phi_{r})=\phi_{\rho}-\omega_{n}$
and the bias term is given by
\[
\omega_{n}=\frac{1}{n}\biggl(\psi\Bigl(\frac{n-1}{2}\Bigl)-\psi\Bigl(\frac{1}{2}\Bigl)\biggl)=\frac{1}{n}\psi\Bigl(\frac{n-1}{2}\Bigl)+\frac{1}{n}(\vartheta_{E}+2\log2),
\]
where $\psi(t)$ is the digamma function and $\vartheta_{E}$ is the
Euler--Mascheroni constant. The bias term $\omega_{n}$ is a non-negative
and non-monotone function of $n$ with $\omega_{2}=0$ and $\omega_{n}\sim\frac{\log n}{n}$
as $n\rightarrow\infty$. For the variance of $\phi_{r}$ we obtain
\[
\mathbb{V}(\phi_{r})=\frac{1}{n^{2}}\biggl(\psi^{\prime}\Bigl(\frac{n-1}{2}\Bigl)+\psi^{\prime}\Bigl(\frac{1}{2}\Bigl)\biggl)=\frac{1}{n^{2}}\biggl(\psi^{\prime}\Bigl(\frac{n-1}{2}\Bigl)+\frac{\pi^{2}}{2}\biggl),
\]
where $\psi^{\prime}(t)=\frac{d}{dt}\psi(t)$ is the trigamma function.
$\square$

\newpage{}
\begin{center}
\includegraphics[scale=0.4]{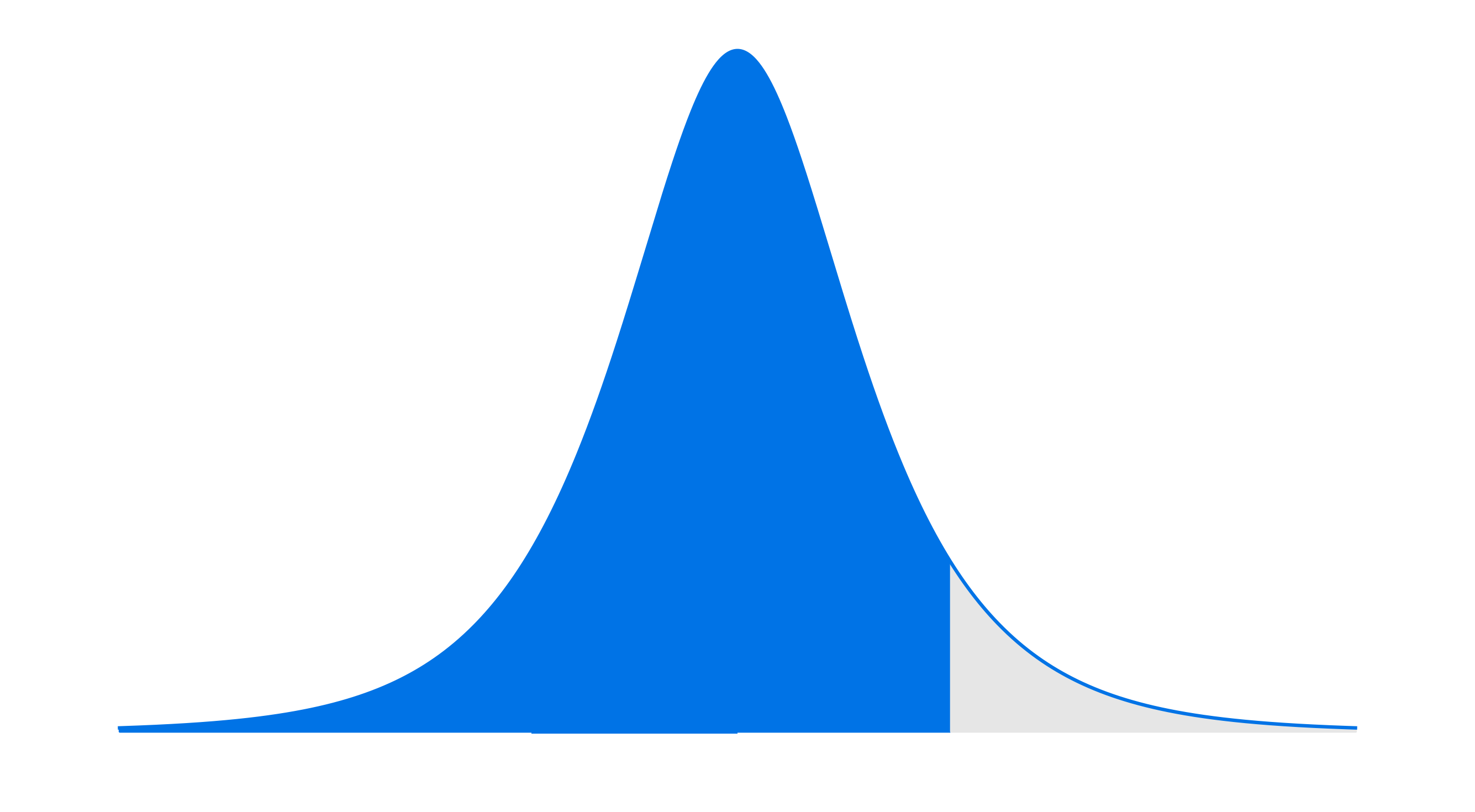}
\par\end{center}

\begin{table}[!h]
\begin{centering}
\bigskip{}

\scriptsize
\begin{tabularx}{\textwidth}{ >{\hsize=1.0\hsize}Y >{\hsize=1.0\hsize}Y  >{\hsize=1.0\hsize}Y  >{\hsize=1.0\hsize}Y  >{\hsize=1.0\hsize}Y  >{\hsize=1.0\hsize}Y  >{\hsize=1.0\hsize}Y  >{\hsize=1.0\hsize}Y  >{\hsize=1.0\hsize}Y  >{\hsize=1.0\hsize}Y  >{\hsize=1.0\hsize}Y  >{\hsize=1.0\hsize}Y  }
\\[-1.0cm]
\toprule
\midrule
  &  \multicolumn{10}{c}{Quantiles}  \\
\cmidrule(l){2-12}
  T  & \multicolumn{1}{r}{0.90}  & \multicolumn{1}{r}{0.95}  & \multicolumn{1}{r}{0.96}  & \multicolumn{1}{r}{0.97}  & \multicolumn{1}{r}{0.975}  & \multicolumn{1}{r}{0.98}  & \multicolumn{1}{r}{0.985}  & \multicolumn{1}{r}{0.99}  & \multicolumn{1}{r}{0.995}  & \multicolumn{1}{r}{0.9975}  & \multicolumn{1}{r}{0.9995}  \\
\midrule
\\[-0.1cm]1 & 1.1731 & 1.6183 & 1.7609 & 1.9444 & 2.0606 & 2.2028 & 2.3860 & 2.6442 & 3.0855 & 3.5268 & 4.5514 \\2 & 1.2210 & 1.6314 & 1.7585 & 1.9197 & 2.0205 & 2.1427 & 2.2984 & 2.5151 & 2.8793 & 3.2373 & 4.0516 \\3 & 1.2391 & 1.6353 & 1.7561 & 1.9082 & 2.0027 & 2.1168 & 2.2613 & 2.4609 & 2.7933 & 3.1168 & 3.8432 \\4 & 1.2488 & 1.6373 & 1.7547 & 1.9018 & 1.9929 & 2.1024 & 2.2407 & 2.4310 & 2.7456 & 3.0497 & 3.7264 \\5 & 1.2549 & 1.6386 & 1.7538 & 1.8978 & 1.9867 & 2.0934 & 2.2277 & 2.4119 & 2.7151 & 3.0067 & 3.6511 \\
\\[-0.12cm]6 & 1.2590 & 1.6394 & 1.7532 & 1.8950 & 1.9824 & 2.0871 & 2.2187 & 2.3987 & 2.6940 & 2.9768 & 3.5983 \\7 & 1.2621 & 1.6401 & 1.7528 & 1.8930 & 1.9793 & 2.0826 & 2.2122 & 2.3890 & 2.6784 & 2.9547 & 3.5592 \\8 & 1.2644 & 1.6406 & 1.7525 & 1.8915 & 1.9770 & 2.0791 & 2.2072 & 2.3817 & 2.6665 & 2.9377 & 3.5289 \\9 & 1.2662 & 1.6410 & 1.7522 & 1.8903 & 1.9751 & 2.0764 & 2.2032 & 2.3758 & 2.6571 & 2.9242 & 3.5049 \\10 & 1.2677 & 1.6414 & 1.7520 & 1.8894 & 1.9736 & 2.0742 & 2.2000 & 2.3711 & 2.6494 & 2.9133 & 3.4853 \\
\\[-0.12cm]11 & 1.2689 & 1.6416 & 1.7519 & 1.8886 & 1.9724 & 2.0724 & 2.1974 & 2.3672 & 2.6431 & 2.9042 & 3.4690 \\12 & 1.2699 & 1.6419 & 1.7518 & 1.8879 & 1.9714 & 2.0708 & 2.1952 & 2.3639 & 2.6377 & 2.8966 & 3.4552 \\13 & 1.2708 & 1.6421 & 1.7517 & 1.8874 & 1.9705 & 2.0696 & 2.1933 & 2.3611 & 2.6332 & 2.8901 & 3.4434 \\14 & 1.2715 & 1.6423 & 1.7516 & 1.8869 & 1.9698 & 2.0684 & 2.1917 & 2.3587 & 2.6292 & 2.8844 & 3.4332 \\15 & 1.2722 & 1.6424 & 1.7515 & 1.8865 & 1.9691 & 2.0675 & 2.1903 & 2.3566 & 2.6258 & 2.8795 & 3.4243 \\
\\[-0.12cm]16 & 1.2727 & 1.6426 & 1.7515 & 1.8861 & 1.9685 & 2.0666 & 2.1890 & 2.3548 & 2.6228 & 2.8752 & 3.4164 \\17 & 1.2732 & 1.6427 & 1.7514 & 1.8858 & 1.9680 & 2.0659 & 2.1880 & 2.3532 & 2.6201 & 2.8713 & 3.4094 \\18 & 1.2737 & 1.6428 & 1.7514 & 1.8855 & 1.9676 & 2.0652 & 2.1870 & 2.3517 & 2.6178 & 2.8679 & 3.4031 \\19 & 1.2741 & 1.6429 & 1.7513 & 1.8853 & 1.9672 & 2.0646 & 2.1861 & 2.3504 & 2.6156 & 2.8648 & 3.3975 \\20 & 1.2745 & 1.6430 & 1.7513 & 1.8851 & 1.9668 & 2.0641 & 2.1853 & 2.3492 & 2.6137 & 2.8620 & 3.3924 \\
\\[-0.12cm]21 & 1.2748 & 1.6431 & 1.7512 & 1.8849 & 1.9665 & 2.0636 & 2.1846 & 2.3482 & 2.6119 & 2.8595 & 3.3878 \\22 & 1.2751 & 1.6431 & 1.7512 & 1.8847 & 1.9662 & 2.0632 & 2.1840 & 2.3472 & 2.6103 & 2.8572 & 3.3836 \\23 & 1.2754 & 1.6432 & 1.7512 & 1.8845 & 1.9659 & 2.0627 & 2.1834 & 2.3463 & 2.6089 & 2.8551 & 3.3797 \\24 & 1.2756 & 1.6433 & 1.7512 & 1.8843 & 1.9657 & 2.0624 & 2.1828 & 2.3455 & 2.6075 & 2.8531 & 3.3761 \\25 & 1.2759 & 1.6433 & 1.7511 & 1.8842 & 1.9655 & 2.0620 & 2.1823 & 2.3447 & 2.6063 & 2.8514 & 3.3728 \\
\\[-0.12cm]30 & 1.2768 & 1.6436 & 1.7511 & 1.8836 & 1.9645 & 2.0607 & 2.1803 & 2.3417 & 2.6013 & 2.8441 & 3.3596 \\35 & 1.2775 & 1.6438 & 1.7510 & 1.8832 & 1.9639 & 2.0597 & 2.1789 & 2.3395 & 2.5977 & 2.8390 & 3.3500 \\40 & 1.2780 & 1.6439 & 1.7510 & 1.8829 & 1.9634 & 2.0589 & 2.1778 & 2.3379 & 2.5950 & 2.8350 & 3.3428 \\45 & 1.2784 & 1.6440 & 1.7509 & 1.8827 & 1.9630 & 2.0584 & 2.1769 & 2.3366 & 2.5929 & 2.8320 & 3.3371 \\50 & 1.2787 & 1.6441 & 1.7509 & 1.8825 & 1.9627 & 2.0579 & 2.1762 & 2.3356 & 2.5912 & 2.8295 & 3.3325 \\55 & 1.2789 & 1.6441 & 1.7509 & 1.8823 & 1.9625 & 2.0575 & 2.1757 & 2.3348 & 2.5899 & 2.8275 & 3.3288 \\
\\[-0.12cm]60 & 1.2792 & 1.6442 & 1.7509 & 1.8822 & 1.9623 & 2.0572 & 2.1752 & 2.3341 & 2.5887 & 2.8258 & 3.3257 \\70 & 1.2795 & 1.6443 & 1.7508 & 1.8820 & 1.9619 & 2.0567 & 2.1745 & 2.3330 & 2.5869 & 2.8232 & 3.3207 \\80 & 1.2798 & 1.6444 & 1.7508 & 1.8819 & 1.9617 & 2.0564 & 2.1739 & 2.3322 & 2.5855 & 2.8212 & 3.3170 \\90 & 1.2800 & 1.6444 & 1.7508 & 1.8817 & 1.9615 & 2.0561 & 2.1735 & 2.3315 & 2.5844 & 2.8196 & 3.3141 \\100 & 1.2801 & 1.6445 & 1.7508 & 1.8816 & 1.9613 & 2.0558 & 2.1732 & 2.3310 & 2.5836 & 2.8184 & 3.3118 \\
\\[-0.12cm]$N(0,1)$ & 1.2816 & 1.6449 & 1.7507 & 1.8808 & 1.9600 & 2.0537 & 2.1701 & 2.3263 & 2.5758 & 2.8070 & 3.2905 \\
\\[-0.2cm]
\midrule
\bottomrule
\end{tabularx}
\par\end{centering}
\caption{Critical values of $\frac{\sqrt{T}(\hat{\gamma}-\phi_{\rho})}{\pi/2}$
for the bivariate elliptical case with homogeneous scale parameters.\label{tab:quantiles_table}}
\end{table}

\end{document}